\documentclass[a4paper,UKenglish]{lipics}

\usepackage[utf8]{inputenc}

\usepackage{etoolbox}
\usepackage{tikz}
\usepackage{hyperref}
\usepackage{amsthm}
\usepackage{amsmath}
\usepackage{amssymb,stmaryrd}
\usepackage{listings}
\usepackage{verbatim}
\usepackage{color}
\usepackage{graphicx}
\usepackage{multirow}
\usepackage{pdfpages}
\usepackage{multibib}
\usepackage{wrapfig}
\usepackage{cutwin}
\usepackage{paralist}
\usepackage{bbding}
\usepackage{relsize}
\usetikzlibrary{shapes}
\usetikzlibrary{positioning}
\usetikzlibrary{arrows}
\usetikzlibrary{calc}
\usetikzlibrary{decorations.text}
\usetikzlibrary{decorations.pathmorphing}
\usetikzlibrary{decorations.pathreplacing}
\usetikzlibrary{automata}
\usetikzlibrary{shapes.callouts}
\usetikzlibrary{shapes.multipart}

\tikzstyle{state}=[minimum size = 6.3mm, inner sep = 0.5mm, circle, draw]
\tikzstyle{statemdp}=[minimum size = 6.2mm, inner sep = 0mm, rounded corners=0, draw,rectangle split, rectangle split horizontal=false, rectangle split parts=2, rectangle split draw splits=false]

\usetikzlibrary{arrows,shapes,positioning,snakes,automata,backgrounds,petri,calc}

\tikzstyle{max}=[thick,draw,minimum size=1.2em,inner sep=0em]
\tikzstyle{min}=[diamond,thick,draw,minimum size=1.4em,%
inner sep=0em]
\tikzstyle{ran}=[circle,thick,draw,minimum size=1.2em,%
inner sep=0em]
\tikzstyle{act}=[circle,thick,draw,fill,minimum size=.7em,%
inner sep=0em]
\tikzstyle{mc}=[rounded corners,thick,draw,minimum size=1.4em,%
inner sep=.5ex]
\tikzstyle{tran}=[thick,draw,->,>=stealth]
\tikzstyle{loop left}=[tran, to path={.. controls +(150:.5)
	and +(210:.5) .. (\tikztotarget) \tikztonodes}]
\tikzstyle{loop right}=[tran, to path={.. controls +(30:1.5)
	and +(330:1.5) .. (\tikztotarget) \tikztonodes}]
\tikzstyle{loop above}=[tran, to path={.. controls +(60:.5)
	and +(120:.5) .. (\tikztotarget) \tikztonodes}]
\tikzstyle{loop below}=[tran, to path={.. controls +(240:.5)
	and +(300:.5) .. (\tikztotarget) \tikztonodes}]

\newcommand{\game}{\mathsf{G}}

\newcommand{\sinit}{s}

\newcommand{\reward}{\rho}
\newcommand{\lr}{\mathit{mp}}
\newcommand{\lmp}{\mathit{lmp}}

\newcommand{\va}{\mathit{va}}
\newcommand{\outcome}[3]{\mathit{outcome}^{#1,#2}_{#3}}
\newcommand{\Nset}{\mathbb{N}}
\newcommand{\Qset}{\mathbb{Q}}

\newcommand{\run}{\lambda}
\newcommand{\init}{\mathit{Init}}
\newcommand{\mem}{\mathit{Mem}}
\newcommand{\up}{\mathit{Up}}
\newcommand{\sel}{\mathit{Const}}
\newcommand{\obj}{\Phi}
\newcommand{\calO}{\mathcal{O}}
\newcommand{\calV}{\mathcal{V}}
\newcommand{\calF}{\mathcal{F}}
\newcommand{\maxr}[1]{\mathit{max}_{\varrho_{#1}}}
\newcommand{\dimr}[1]{\mathit{dim}_{\varrho_{#1}}}
\newcommand{\pl}{\oplus}
\newcommand{\Succ}{\mathit{succ}}
\newcommand{\gfix}{\mathit{gfix}}
\newcommand{\reach}{\mathit{Reach}}

\newcommand{\mmod}{\mathrm{\ mod\ }}

\newcommand{\NP}{\mathsf{NP}}
\newcommand{\coNP}{\mathsf{coNP}}
\newcommand{\PSPACE}{\mathsf{PSPACE}}

\newcommand{\calT}{\mathcal{T}}
\newcommand{\Mod}[1]{\ (\text{mod}\ #1)}
\newcommand{\EXPTIME}{\mathsf{EXPTIME}}
\newcommand{\PTIME}{\mathsf{P}}

\tikzstyle{ran}=[circle,thick,draw,minimum size=2ex,inner sep=0em]

\tikzstyle{ran}=[circle,thick,draw,minimum size=2ex,inner sep=0em]

\allowdisplaybreaks

\newcommand{\vojta}[1]{}%

\newtheorem{proposition}[theorem]{\bfseries  Proposition}
\newtheorem{claim}[theorem]{\bfseries Claim}

\newcommand{\theoremlike}[2]{\par\medskip\penalty-250\refstepcounter{theorem}{{\bfseries\noindent#2
			\ref{#1}.}}}
\newcommand{\thmhelperpre}[2]{\theoremlike{#1}{#2}}
\newcommand{\thmhelperpost}{\par\medskip}

\newenvironment{reftheorem}[1]{\thmhelperpre{#1}{Theorem}}{\thmhelperpost}
\newenvironment{reflemma}[1]{\thmhelperpre{#1}{Lemma}}{\thmhelperpost}
\newenvironment{refproposition}[1]{\thmhelperpre{#1}{Proposition}}{\thmhelperpost}

\begin{document}
\title{Stability in Graphs and Games}

\author[1]{Tom\'a\v{s} Br\'azdil}
\author[2]{Vojt\v{e}ch~Forejt}
\author[1]{Anton\'in Ku\v{c}era}
\author[3]{Petr Novotn\'y}
\affil[1]{Faculty of Informatics, Masaryk University, Czech Republic}
\affil[2]{Department of Computer Science, University of Oxford, UK}
\affil[3]{IST Austria, Klosterneuburg, Austria}

\maketitle

\begin{abstract}
We study graphs and two-player games in which rewards are assigned to states,
and the goal of the players is to satisfy or dissatisfy certain property
of the generated outcome, given as a mean payoff property. 
Since the notion of mean-payoff does not reflect possible fluctuations from the mean-payoff along a run, we propose definitions and algorithms for capturing the
\emph{stability} of the system, and give algorithms for deciding if a given mean payoff
and stability objective can be ensured in the system.
\end{abstract}

\section{Introduction}
\label{sec-intro}

Finite-state graphs and games are used in formal verification as foundational models that capture behaviours of systems
with controllable decisions and possibly also the presence of adversarial environment. States correspond to possible configurations
of a system, and edges describe how configurations can change. In a game, each state is owned by one of two players, and
the player owning the state decides what edge will be taken. A graph is a game where only one of the players is present.
When the choice of the edges is resolved, we obtain an {\em outcome} which is an infinite sequence of states and edges describing
the execution of the system.

The long-run average performance of a run is measured by the associated \emph{mean-payoff}, which is the limit average reward per visited state along the run. It is well known that memoryless deterministic strategies suffice to optimise the mean payoff, and the corresponding decision problem is in $\NP\cap\coNP$ for games and in $\PTIME$ for graphs. If the rewards assigned to the states are multi-dimensional vectors of numbers, then the problem becomes $\coNP$-hard for games~\cite{Velner2015177}.\vojta{check the reference. Also, what about graphs?}

Although the mean payoff provides an important metric for the average behaviour of the system, by definition it neglects all information about the fluctuations from the mean payoff along the run. For example, a ``fully stable''  run where the associated sequence of rewards is $1,1,1,1,\ldots$ has the same mean payoff (equal to $1$) as a run producing $n,0,0,\ldots,n,0,0,\ldots$ where a state with 
the reward $n$ is visited once in $n$ transitions. In many situations, the first run is much more desirable that the second one. Consider, e.g., a video streaming application which needs to achieve a sufficiently high bitrate (a long-run average number of bits delivered per second) but, in addition, a sufficient level of ``stability'' to prevent buffer underflows and overflows which would cause data loss and stuttering. Similar problems appear also in other contexts. For example, production lines should be not only efficient (i.e., produce the number of items per time unit as high as possible), but also ``stable'' so that the available stores are not overfilled and there is no ``periodic shortage'' of the produced items. A food production system should not only produce a sufficiently large amount of food per day on average, but also a certain amount of food daily. These and similar problems motivate the search for a suitable formal notion capturing the intuitive understanding of ``stability'', and developing algorithms that can optimize the performance under given stability constraints. That is, we are still seeking for a strategy optimizing the mean payoff, but the search space is restricted to the subset of all strategies that achieve a given stability constraint.

Since the mean-payoff $\lr(\lambda)$ of a given run $\lambda$ can be seen as the average reward of a state visited along~$\lambda$, a natural idea is to define the stability of~$\lambda$ as \emph{sample variance} of the reward assigned to a state along~$\lambda$. More precisely, let $r_i$ be the reward if the \mbox{$i$-th} state visited by $\lambda$, and let $c_0,c_1,\ldots$ be an infinite sequence where $c_i = (\lr(\lambda) - r_i))^2$. The \emph{long-run variance} of the reward assigned to a state along $\lambda$, denoted by $\va(\lambda)$, is the limit-average of $c_0,c_1,\ldots$ The notion of long-run variance has been introduced and studied for Markov decision processes in \cite{BCFK13}. If $\va(\lambda)$ is small, then large fluctuations from $\lr(\lambda)$ are rare. Hence, if we require that a strategy should optimize mean-payoff while keeping the long-run variance below a given threshold, we in fact impose a \emph{soft} stability constraint which guarantees that ``bad things do not happen too often''. This may or may not be sufficient.

In this paper, we are particularly interested in formalizing \emph{hard} stability constraints which guarantee that ``bad things never happen''. We introduce a new type of objectives called \emph{window-stability multi-objectives} that can express a rich set of hard stability constraints, and we show that the set of \emph{all} strategies that achieve a given window-stability multi-objective can be characterized by an effectively constructible \emph{finite-memory permissive strategy scheme}. From this we obtain a meta-theorem saying that if an objective (such as mean-payoff optimization) is solvable for finite-state games (or graphs), then the same objective is solvable also under a given window-stability multi-objective constraint. We also provide the associated upper and lower complexity bounds demonstrating that the time complexity of our algorithms is ``essentially optimal''.

More specifically, a single \emph{window-stability objective} (inspired by \cite{raskin-windows}, see Related work below) is specified by a window length $W \geq 1$, a checkpoint distance $D \geq 1$, and two bounds $\mu$ and $\nu$. For technical reasons, we assume that $D$ divides $W$. Every run $\lambda = s_0,s_1,s_2,\ldots$ then contains infinitely many \emph{checkpoints}  $s_0,s_{D},s_{2D},s_{3D},\ldots$ The objective requires that the average reward assigned to the states $s_j,\ldots, s_{j+W-1}$, where $s_j$ is a checkpoint, is between $\mu$ and $\nu$. In other words, the ``local mean-payoff'' computed for the states fitting into a window of length $W$ starting at a checkpoint must be within the ``acceptable'' bounds $\mu$ and $\nu$. The role of $W$ is clear, and the intuition behind $D$ is the following. Since $D$ divides $W$, there are two extreme cases: $D=1$ and $D=W$. For $D=1$, the objective closely resembles the standard ``sliding window'' model over data streams~\cite{DGIM:data-streams-sliding-window}; we require that the local mean-payoff stays within the acceptable bounds ``continuously'', like the ``local bitrate'' in video-streaming. If $D=W$, then the windows do not overlap at all. This is useful in situations when we wish to guarantee some time-bounded periodic progress. For example, if we wish to say that the number of items produced per day stays within given bounds, we set $W$ so that it represents the (discrete) time of one day and put $D=W$. However, there can be also scenarios when we wish to check the local mean-payoff more often then once during $W$ transitions, but not ``completely continuously''. In these cases, we set $D$ to some other divisor of~$W$. A \emph{window-stability multi-objective} is a finite conjunction of single window-stability objectives, each with dedicated rewards and parameters. Hence, window-stability multi-objectives allow for capturing more delicate stability requirements such as  ``a factory should produce between 1500 and 1800 gadgets every week, and in addition, within every one-hour period at least 50 computer chips are produced, and in addition, the total amount of waste produced in 12 consecutive hours does not exceed 500~kg.''

\smallskip
\noindent
\textbf{Our contribution} can be summarized as follows:
\smallskip

\noindent
(A) We introduce the concept of window-stability multi-objectives. 
\smallskip

\noindent
(B) We show that there is an algorithm which inputs a game $\game$ and a window-stability multi-objective $\Delta$, and outputs a \emph{finite-state permissive strategy scheme} for $\Delta$ and $\game$. A finite-state permissive strategy scheme for $\Delta$ and $\game$ is a finite-state automaton $\Gamma$ which reads the history of a play and constraints the moves of Player~$\Box$ (who aims at satisfying~$\Delta$) so that a strategy $\sigma$ achieves $\Delta$ in $\game$ iff $\sigma$ is admitted by $\Gamma$. Hence, we can also compute a synchronized product $\game \times \Gamma$ which is another game where the set of \emph{all} strategies for Player~$\Box$ precisely represents the set of all strategies for Player~$\Box$ in $\game$ which achieve the objective $\Delta$. Consequently, \emph{any} objective of the form $\Delta \wedge \Psi$ can be solved for $\game$ by solving the objective $\Psi$ for $\game \times \Gamma$. In particular, this is applicable to mean-payoff objectives, and thus we solve the problem of optimizing the mean-payoff under a given window-stability multi-objective constraint. We also analyze the time complexity of these algorithms, which reveals that the crucial parameter which negatively influences the time complexity is the number of checkpoints in a window (i.e., $W/D$).
\smallskip

\noindent
(C) We complement the upper complexity bounds of the previous item by lower complexity bounds that indicate that the time complexity of our algorithms is ``essentially optimal''. Some of these results follow immediately from existing works~\cite{Velner2015177,raskin-windows}. The main contribution is the result which says that solving a (single) window-stability objective is $\PSPACE$-hard for games and $\NP$-hard for graphs, even if all numerical parameters ($W$, $D$, $\mu$, $\nu$, and the rewards) are encoded in \emph{unary}. The proof is based on novel techniques and reveals that the number of checkpoints in a window (i.e., $W/D$) is a crucial parameter which makes the problem computationally hard. The window stability objective constructed in the proof satisfies $D=1$, and the tight window overlapping is used to enforce a certain consistency in Player~$\Box$ strategies.
\smallskip

\noindent
(D) For variance-stability, we argue that while it is natural in terms of using standard mathematical definitions,
it does not prevent unstable behaviours. In particular, we show that the variance-stability objective may demand an infinite-memory strategy which switches between two completely different modes of behaviour with smaller and smaller frequency. We also show that the associated variance-stability problem with single-dimensional rewards is in $\NP$ for graphs. For this we use some of the results from~\cite{BCFK13} where the variance-stability is studied in the context of Markov decision processes. The main difficulty is a translation from randomized stochastic-update strategies used in~\cite{BCFK13} to deterministic strategies.

\smallskip
\noindent
\textbf{Related work.}
Multi-dimensional mean-payoff games were studied in~\cite{Velner2015177}, where it was shown that the lim-inf problem, relevant to our setting,
is $\coNP$-hard. Further,~\cite{CRR-acta14}
studies memory requirements for the objectives, and~\cite{Velner2015} shows that for a ``robust''
extension (where Boolean combinations of bounds on the resulting vector of mean-payoffs are allowed) the problem becomes undecidable. Games with quantitative objectives in which both lower
and upper bound on the target value of mean-payoff is given were studied in~\cite{DBLP:conf/fsttcs/HunterR14}. We differ from these approaches by requiring the “interval” bounds to be satisfied within finite windows, making our techniques and results very different.

As discussed above, we rely on the concept of \emph{windows}, which was in the synthesis setting studied in~\cite{raskin-windows} (see also~\cite{HPR:foggy-windows}), as a conservative approximation of the standard mean-payoff objective. More concretely, the objectives in $\cite{raskin-windows}$ are specified by a maximal window length $W$ and a threshold $t$. The task is to find a strategy that achieves the following property of runs: a run can be partitioned into contiguous windows of length \emph{at most} $W$ such that in each window, the reward accumulated inside the window divided by the window length is at least~$t$. The objective ensures a local progress in accumulating the reward, and it was not motivated by capturing stability constraints. There are also technical differences in solving these objectives and in the associated complexity bounds, mainly due to the absence of window overlapping (in particular, the $\PSPACE$ lower bound discussed in the point~(C) above does not carry over to the setting of~\cite{raskin-windows}).  

The notion of finite-state permissive strategy scheme is based on the concept of \emph{permissive strategies} \cite{DBLP:journals/ita/BernetJW02} and multi-strategies \cite{Bouyer2011,patricia-multistrategies}. 

The notion of long-run variance has been introduced and studied for Markov decision processes in~\cite{BCFK13}.  Since we consider deterministic strategies, none of our results is a special case of~\cite{BCFK13}, and we have to overcome new difficulties as it is explained in Section~\ref{sec-results-variance}.

More generally, our paper fits into an active field of multi-objective strategy synthesis, where some objectives capture the ``hard'' constraints and the other ``soft'', often quantitative, objectives. Examples of recent results in this area include~\cite{Bohy:2013:SLS:2450387.2450406}, where a 2-$\EXPTIME$ algorithm is given for the synthesis of combined  LTL and mean-payoff objectives, \cite{DBLP:conf/lics/ChatterjeeHJ05}, where a combination of parity and mean-payoff performance objectives is studied, or \cite{CHL:obliging-games}, where the controlling player must satisfy a given $\omega$-regular objective while allowing the adversary to satisfy another ``environmental'' objective. 

\section{Preliminaries}
\label{sec:prelims}

We use $\Nset$, $\Nset_0$, and $\Qset$ to denote the sets of positive integers, non-negative integers, and rationals, respectively. Given a set $M$, we use $M^*$ to denote the set of all finite sequences (words) over~$M$, including the empty sequence. 

A \emph{game} is a tuple $\game = (S, (S_\Box,S_\Diamond),E)$ where $S$ is a non-empty set of \emph{states}, $(S_\Box,S_\Diamond)$ is a partition of $S$ into two subsets controlled by Player~$\Box$ and Player~$\Diamond$, respectively, and $E \subseteq S \times S$ are the \emph{edges} of the game such that for every $s \in S$ there is at least one edge $(s,t) \in E$. A \emph{graph} is a game such that $S_\Diamond = \emptyset$.
A \emph{run} in $\game$ is an infinite path in the underlying directed graph of $\game$. An \emph{objective} is a Borel property\footnote{Recall that the set of all runs can be given the standard Cantor topology. A property is Borel if the set of all runs satisfying the property belongs to the  $\sigma$-algebra generated by all open sets in this topology.} of runs. Note that the class of all objectives is closed under conjunction.

A \emph{strategy} for player~$\odot$, where $\odot \in \{\Box,\Diamond\}$ is a function $\tau : S^* S_{\odot} \rightarrow S$ satisfying that $(s,\sigma(hs))\in E$ for all $s\in S_{\odot}$ and $h\in S^*$. The sets of all strategies of Player~$\Box$ and Player~$\Diamond$ are denoted by $\Sigma_\game$ and $\Pi_\game$, respectively. When $\game$ is understood, we write just $\Sigma$ and $\Pi$.  A pair of strategies $(\sigma,\pi) \in \Sigma \times \Pi$ together with an initial state $\sinit$ induce a unique run $\outcome{\sigma}{\pi}{\sinit}$ in the standard way. We say that a strategy $\sigma \in \Sigma$ \emph{achieves} an objective $\obj$ in a state~$s$ if $\outcome{\sigma}{\pi}{\sinit}$ satisfies $\obj$ for every $\pi \in \Pi$. The set of all $\sigma \in \Sigma$ that achieve $\obj$ in $s$ is denoted by $\Sigma^{\obj}(s)$. An objective $\obj$ is \emph{solvable} for a given subclass $\mathcal{G}$ of finite-state games if there is an algorithm which inputs $\game \in \mathcal{G}$ and its state $s$, and decides whether $\Sigma^{\obj}(s) = \emptyset$. If $\Sigma^{\obj}(s) \neq \emptyset$, then the algorithm also outputs a 
(finite description of) $\sigma \in \Sigma^{\obj}(s)$.

We often consider strategies of Player~$\Box$ tailored for a specific initial state. A finite sequence of states $s_0,\ldots,s_n$ is \emph{consistent} with a given $\sigma \in \Sigma$ if $s_0,\ldots,s_n$ is a finite path in the graph of $\game$, and $\sigma(s_0,\ldots,s_{i}) = s_{i+1}$ for every $0 \leq i < n$ where $s_i \in V_\Box$. Given $\sigma,\sigma' \in \Sigma$ and $\sinit \in S$, we say that $\sigma$ and $\sigma'$ are \emph{$\sinit$-equivalent}, written $\sigma \equiv_{\sinit} \sigma'$, if $\sigma$ and $\sigma'$
agree on all finite sequences of states initiated in $\sinit$ that are consistent with $\sigma$. Note that if 
$\sigma \equiv_{\sinit} \sigma'$, then $\outcome{\sigma}{\pi}{\sinit} = \outcome{\sigma'}{\pi}{\sinit}$ for every $\pi \in \Pi$. 

A \emph{reward function} $\varrho: S \rightarrow \Nset_0^k$, where $k \in \Nset$, assigns non-negative integer vectors to the states of~$\game$. We use $\dimr{}$ to denote the dimension $k$ of $\varrho$, and $\maxr{}$ to denote the maximal number employed by $\varrho$, i.e., $\maxr{} = \max \{\varrho(s)[i] \mid 1 \leq i \leq k, s \in S\}$. An objective is \emph{reward-based} if its defining property depends just on the sequence of rewards assigned to the states visited by a run. 

For every run $\run=s_0,s_1,\ldots$ of $\game$, let
$
   \lr_\varrho(\run) = \liminf_{n\rightarrow \infty} \frac{1}{n+1} \sum_{i=0}^n \varrho(s_i)
$
be the \emph{mean payoff} of $\run$, where the $\liminf_{n\rightarrow \infty}$ is taken component-wise.
A \emph{mean-payoff} objective is a pair $(\varrho,\kappa)$, where $\varrho : S \rightarrow \Nset_0^k$ is a reward function and $b \in \Qset^k$. A run $\run$ satisfies a mean-payoff objective $(\varrho,b)$ if $\lr_\varrho(\run) \geq b$.

Similarly, the \emph{long-run variance of the reward} of a run $\run$ is defined by $\va_\varrho(\run) = \limsup_{n\rightarrow \infty} \frac{1}{n+1} \sum_{i=0}^n (\reward(s_i) - \lr(run))^2$; intuitively, the long-run variance is a limit inferior of sample variances where the samples represent longer and longer run prefixes. A \emph{variance-stability} objective is a triple $(\varrho,b,c)$, where $\varrho : S \rightarrow \Nset_0^k$ is a reward function and $b,c \in \Qset^k$.  A run $\run$ satisfies a variance-stability objective $(\varrho,b,c)$ if $\lr_\varrho(\run) \geq b$ and  $\va_\varrho(\run) \leq c$.

Let $W \in \Nset$ be a \emph{window size} and $D \in \Nset$ a \emph{checkpoint distance} such that $D$ divides~$W$. For every $\ell \in \Nset_0$, the \emph{local mean payoff at the $\ell^{\mathit{th}}$ checkpoint in a run $\run$} is defined by
$\lmp_{W,D,\varrho,\ell}(\run)  =  \frac{1}{W} \sum_{i=0}^{W{-}1} \varrho(s_{\ell\cdot D+i})$
where $(\vec{v}/a)[i] = \vec{v}[i]/a$. 
Thus, every run $\run$ determines the associated infinite sequence $\lmp_{W,D,\varrho,0}(\run),\lmp_{W,D,\varrho,1}(\run),\lmp_{W,D,\varrho,2}(\run),\dots$ of local mean payoffs.
A \emph{window-stability} objective is a tuple $\obj = (W,D,\varrho,\mu,\nu)$, where $W,D \in \Nset$ such that $D$ divides $W$, $\varrho : S \rightarrow \Nset_0^k$ is a reward function, and $\mu,\nu \in \Qset^k$. A run $\run$ satisfies $\obj$ if, for all $\ell \in \Nset$, we have that $\mu \leq \lmp_{W,D,\varrho,\ell}(\run) \leq \nu$. A \emph{window-stability multi-objective} is a finite conjunction of window-stability objectives.

In this paper, we study the solvability of variance-stability objectives, window-stability multi-objectives, and objectives of the form $\Delta \wedge \Psi$ where $\Delta$ is a window-stability multi-objective and $\Psi$ a mean-payoff objective.

\section{The Window-Stability Multi-Objectives}
\label{sec-results-window}

This section is devoted to the window-stability multi-objectives and objectives of the form $\Delta \wedge \Psi$, where $\Delta$ is a window-stability multi-objective. In Section~\ref{sec-algorithm}, we show how to solve these objectives for finite-state games, and we derive the corresponding upper complexity bounds. The crucial parameter which makes the problem computationally hard is the number of checkpoints in a window. In Section~\ref{sec-hardness}, we show that this blowup is unavoidable assuming the expected relationship among the basic complexity classes.

\subsection{Solving Games with Window-Stability Multi-Objectives}
\label{sec-algorithm}

We start by recalling the concept of \emph{most permissive strategies} which was introduced in \cite{DBLP:journals/ita/BernetJW02}. Technically, we define \emph{permissive strategy schemes} which suit better our needs, but the underlying idea is the same.

\begin{definition}
\label{def-delector}
    Let $\game = (S, (S_\Box,S_\Diamond),E)$ be a game. A \emph{(finite-memory) strategy scheme} for $\game$ is a tuple $\Gamma = (\mem,\up,\sel,\init)$, where $\mem \neq \emptyset$ is a finite set of \emph{memory elements}, $\up : S \times \mem \rightarrow \mem$ is a \emph{memory update} function, $\sel : S_\Box \times \mem \rightarrow 2^{S}$ is a \emph{constrainer} such that $\sel(s,m) \subseteq \{s' \in S \mid (s,s') \in E\}$, and $\init : S \rightharpoonup M$ is a partial function assigning initial memory elements to some states of~$S$.
	
	We require\footnote{Alternatively, we could stipulate $\sel(s,m) \neq \emptyset$ for all $(s,m) \in S_\Box \times \mem$, but this would lead to technical complications in some proofs. The presented variant seems slightly more convenient.} that $\sel(s,m) \neq \emptyset$ for all $(s,m) \in \reach(\init)$ such that $s \in S_\Box$. Here, $\reach(\init)$ is the least fixed-point of $\calF : 2^{S \times \mem} \rightarrow 2^{S \times \mem}$ where $\calF(\Omega)$ consists of all $(s',m')$ such that $(s',m') \in \init$, or there is some $(s'',m'') \in \Omega$ such that $(s'',s') \in E$ and $\up(s'',m'') = m'$; if $s'' \in S_\Box$, we further require $s' \in \sel(s'',m'')$.
	\qed
\end{definition}

We say that $\Gamma$ is \emph{memoryless} if $\mem$ is a singleton. Every strategy scheme $\Gamma = (\mem,\up,\sel,\init)$ for a game $\game = (S, (S_\Box,S_\Diamond),E)$ determines a  game $\game_\Gamma = (S {\times} \mem,(S_\Box {\times} \mem,S_\Diamond {\times} \mem),F)$, where 
\begin{itemize}
\item for every $(s,m) \in S_\Diamond {\times} \mem$, $((s,m),(s',m')) \in F$ iff $\up(s,m) = m'$ and $(s,s') \in E$;
\item for every $(s,m) \in S_\Box {\times} \mem$ where $\sel(s,m) \neq \emptyset$, we have that $((s,m),(s',m')) \in F$ iff $\up(s,m) = m'$ and $(s,s') \in \sel(s,m)$;
\item for every $(s,m) \in S_\Box {\times} \mem$ where $\sel(s,m) = \emptyset$, we have that $((s,m),(s',m')) \in F$ iff $s = s'$ and $m=m'$.
\end{itemize}  

A strategy $\sigma \in \Sigma_\game$ is \emph{admitted} by $\Gamma$ in a given $s \in S$ if $\init(s) \neq {\perp}$ and for every finite path $s_0,\ldots,s_n$ in $\game$ initiated in $s$ which is consistent with $\sigma$ there is a finite path  $(s_0,m_0),\ldots,(s_n,m_n)$ in $\game_\Gamma$ such that $m_0 = \init(s_0)$ and $s_{i+1} \in \sel(s_i,m_i)$ for all \mbox{$0 \leq i <n$} where $s_i \in S_\Box$. Observe that if $\sigma$ is admitted by $\Gamma$ in $s$, then $\sigma$ naturally induces a strategy $\tau[\sigma,s] \in \Sigma_{\game_\Gamma}$ which is unique up to $\equiv_{(s_0,m_0)}$. Conversely, every $\tau \in \Sigma_{\game_\Gamma}$ and every $s \in S$ where $\init(s) \neq {\perp}$ induce a strategy $\sigma[\tau,s] \in \Sigma_\game$ such that, for every finite path $(s_0,m_0),\ldots,(s_n,m_n)$ initiated in $(s,\init(s))$ which is consistent with $\tau$, we have that 
\(
  \sigma[\tau,s](s_0,\ldots,s_n) = s_{n+1}  \mbox{ iff } \tau((s_0,m_0),\ldots,(s_n,m_n))=(s_{n+1},m_{n+1}) \,. 
\)
Note that $\sigma[\tau,s]$ is determined uniquely up to $\equiv_s$. 
\begin{definition}
   Let $\game$ be a game, $\Gamma$ a strategy scheme for $\game$, $\Lambda_\game \subseteq \Sigma_\game$, $\Lambda_{\game_\Gamma} \subseteq \Sigma_{\game_\Gamma}$, and $s \in S$. We write $\Lambda_\game \approx_s \Lambda_{\game_\Gamma}$ if the following conditions are satisfied:
   \begin{itemize}  
	   \item Every $\sigma \in \Lambda_\game$ is admitted by $\Gamma$ in $s$, and there is $\tau \in \Lambda_{\game_\Gamma}$ such that $\tau[\sigma,s] \equiv_{(s,\init(s))} \tau$.
	   \item For every $\tau \in \Lambda_{\game_\Gamma}$ there is $\sigma \in \Lambda_\game$ such that $\sigma[\tau,s] \equiv_s \sigma$.
   \end{itemize}

   \noindent
   Further, we say that $\Gamma$ is \emph{permissive} for an objective~$\obj$ if $\Sigma^{\obj}_{\game}(s) \approx_s \Sigma_{\game_\Gamma}(s)$ for all $s \in S$, where $\Sigma_{\game_\Gamma}(s)$ is either $\emptyset$ or $\Sigma_{\game_\Gamma}$, depending on whether $\init(s) = {\perp}$ or not, respectively.
\end{definition}

The next proposition follows immediately. 
\begin{proposition}
	\label{prop-conj}
	Let $\game$ be a game, $\obj,\Psi$ objectives, and $\Gamma$ a strategy scheme permissive for~$\obj$. Then, for every $s \in S$ we have that $\Sigma^{\obj \wedge \Psi}_{\game}(s) \approx_s \Sigma^{\Psi}_{\game_\Gamma}(s)$.
\end{proposition}

Another simple but useful observation is that the class of objectives for which a permissive strategy scheme exists is closed under conjunction. 

\begin{proposition}
	\label{prop-sel-conj}
	Let $\game = (S, (S_\Box,S_\Diamond),E)$ be a finite-state game, and $n \in \Nset$. Further, for every $1 \leq i \leq n$, let $\Gamma_i = (\mem_i,\up_i,\sel_i,\init_i)$ be a strategy scheme for $\game$ which is permissive for $\obj_i$. Then there is a strategy scheme for $\game$ with $\prod_{i=1}^n |\mem_i|$ memory elements computable in $\calO(|S|^2 \cdot |E| \cdot \prod_{i=1}^n |\mem_i|^2)$ time which is permissive for $\obj_1 \wedge \cdots \wedge \obj_n$.
\end{proposition}

As it was noted in \cite{DBLP:journals/ita/BernetJW02}, permissive strategy schemes do not exist for objectives which admit non-winning infinite runs that do not leave the winning region of player~$\Box$, such as reachability, B\"{u}chi, parity, mean-payoff, etc. On the other hand, permissive strategy schemes exists for ``time bounded'' variants of these objectives. Now we show how to compute a permissive strategy scheme for a given window-stability objective.

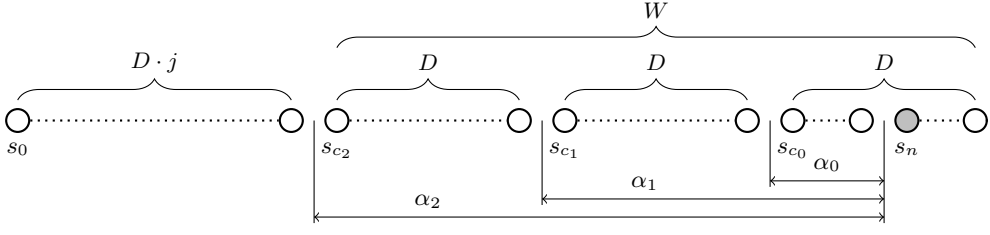
\begin{figure}
\begin{center}
	\begin{tikzpicture}[x=.6mm,y=.8mm,font=\small]
	\foreach \i in {0,60,70,110,120,160,170,185,210}{%
	   \node (s\i) at (\i,0) [ran] {};	
    }
    \node (s195) at (195,0) [ran,fill=lightgray] {};
	\foreach \i/\j in {0/60,70/110,120/160,170/185,195/210}{%
	   \draw [dotted,thick] (s\i) -- (s\j);
	}
	\draw [dotted,thick] (s210) -- +(1.5,0);
	\node [below] at (s0.south)  {$s_0$};
	\node [below] at (s195.south) {$s_n$};
	\node [below] at (s170.south) {$s_{c_0}$};
	\node [below] at (s120.south) {$s_{c_1}$};
	\node [below] at (s70.south) {$s_{c_2}$};
	\foreach \i/\j/\l in {0/60/$D\cdot j$,70/110/$D$,120/160/$D$,170/210/$D$}{%
	    \draw [decorate,decoration={brace,amplitude=8pt,raise=3pt}] (s\i.north) -- (s\j.north) node [midway,yshift=6mm] {\footnotesize \l};
	}
	\draw [decorate,decoration={brace,amplitude=8pt,raise=8mm}] (s70.north) -- (s210.north) node [midway,yshift=13mm] {\footnotesize $W$};
	\draw [] (165,0) -- +(0,-11);	
	\draw [] (115,0) -- +(0,-14);	
	\draw [] (65,0)  -- +(0,-17);
	\draw [] (190,0) -- +(0,-17);
	\draw [<->] (165,-10) -- node[above] {$\alpha_0$} (190,-10);
	\draw [<->] (115,-13) -- node[pos=0.3,above] {$\alpha_1$}(190,-13);
	\draw [<->] (65,-16) -- node[pos=0.2,above] {$\alpha_2$}(190,-16);
	\end{tikzpicture}
\end{center}
\caption{The information represented by the memory elements of $\Gamma$ (for $\ell =3$).}
\label{fig-selector-Gamma}
\end{figure}
\begin{theorem}
\label{thm-selector}
    Let $\game = (S, (S_\Box,S_\Diamond),E)$ be a finite-state game and $\Phi = (W,D,\varrho,\mu,\nu)$ a window-stability objective where $\dimr{} = k$. Then there is a strategy scheme $\Gamma$ with \mbox{$W \cdot (\maxr{} \cdot W)^{k\cdot (W/D)}$} memory elements computable in $\calO(|S|^2 \cdot |E| \cdot W^2 \cdot ( \maxr{}\cdot W)^{2k\cdot (W/D)})$ time which is permissive for~$\Phi$.
\end{theorem}

\begin{proof}
Let $\ell = W/D$  and $\calV = \{0,\ldots,\maxr{} \cdot (W{-}1)\}^k$. We put  
\begin{itemize}
	\item $\mem = \{0,\ldots,D{-}1\} \times \{0,\ldots,\ell{-}1\} \times \calV^\ell$.
\end{itemize}

\noindent
Our aim is to construct $\Gamma$ so that for every run $s_0,s_1,\ldots$ in $\game$, the memory elements in the corresponding run $(s_0,m_0),(s_1,m_1),\ldots$ in $\game_\Gamma$, where $(s_0,m_0) \in \init$, satisfy the following. Let $n \in \Nset_0$, and let $m_n = (i,j,\alpha_0,\ldots,\alpha_{\ell-1})$. Then 
\begin{itemize}
	\item $i = n \mmod D$ is the number of steps since the last checkpoint, and $j = \min\{ \lfloor n/D \rfloor, \ell{-}1\}$ is a bounded counter which stores the number of checkpoint visited, up to $\ell-1$
	 (this information is important for the initial $W$ steps);
	\item for every $0 \leq r < \ell$, we put $c_r = n - r \cdot D - (n \mmod D)$ if $n - r \cdot D - (n \mmod D) \geq 0$, otherwise $c_r = n$.
	Intuitively, the state $s_{c_r}$ is the $r$-th previous checkpoint visited along $s_0,s_1,\ldots$ before visiting the state $s_n$ (see Figure~\ref{fig-selector-Gamma}). If the total number of checkpoints visited along the run up to $s_n$ (including $s_n$) is less than $r$, we put $c_r = n$. The vector $\alpha_r$ stored in $m_n$ is then equal to the total reward accumulated between $s_{c_r}$ and $s_n$ (not including~$s_n$), i.e., $\alpha_r = \sum_{t=c_r}^{n-1} \varrho(s_t)$ where the empty sum is equal to $\vec{0}$. In particular $m_0 = (0,0,\vec{0},\ldots,\vec{0})$.
\end{itemize} 

\noindent
Note that by Definition~\ref{def-delector}, we are obliged to define $\up(s,m)$ for \emph{all} pairs $(s,m) \in S \times \mem$, including those that will not be reachable in the end. Let `$\pl$' be a bounded addition over $\Nset_0$ defined by $a \pl b = \min\{a+b,\maxr{} \cdot (W{-}1)\}$. We extend `$\pl$' to $\calV$ in the natural (component-wise) way. The function $\up$ is constructed as follows (consistently with the above intuition): 
\begin{itemize}
\item For all $i,j \in  \Nset_0$ such that $0 \leq i \leq D-2$ and  $0 \leq j \leq \ell-1$, we put   $\up(s,(i,j,\alpha_0,\ldots,\alpha_{\ell-1})) = (i{+}1,j,\alpha_0 \pl \varrho(s),\ldots,\alpha_j \pl \varrho(s),\alpha_{j+1},\ldots,\alpha_{\ell-1})$.	
\item For all $j \in  \Nset_0$ such that $0 \leq j \leq \ell-2$, we put   $\up(s,(D{-}1,j,\alpha_0,\ldots,\alpha_{\ell-1})) = (0,j{+}1,\vec{0},\alpha_0 \pl \varrho(s),\ldots,\alpha_j \pl \varrho(s),\alpha_{j+1},\ldots,\alpha_{\ell-2})$.	
\item $\up(s,(D{-}1,\ell{-}1,\alpha_0,\ldots,\alpha_{\ell-1})) = (0,\ell{-1},\vec{0},\alpha_0 \pl \varrho(s),\ldots,\alpha_{\ell-2} \pl \varrho(s))$.	 
\end{itemize}

\noindent
For every $(s,m) \in S \times \mem$, let $\Succ(s,m)$ be the set of all $(s',m') \in S \times \mem$ such that  $(s,s') \in E$ and $\up(s,m) = m'$.
Now we define a function $\calF : 2^{S \times \mem} \rightarrow 2^{S \times \mem}$ such that, for a given $\Omega \subseteq S \times \mem$, the set $\calF(\Omega)$ consists of all $(s,(i,j,\alpha_0,\ldots,\alpha_{\ell-1}))$ satisfying the following conditions:
\begin{itemize}
	\item if $i = D{-}1$ and $j = \ell{-}1$, then $\mu\cdot W \leq \alpha_{\ell-1} + \varrho(s) \leq \nu\cdot W$.%
	\item if $s \in S_\Diamond$, then $\Succ(s,(i,j,\alpha_0,\ldots,\alpha_{\ell-1})) \subseteq \Omega$.
	\item if $s \in S_\Box$, then $\Succ(s,(i,j,\alpha_0,\ldots,\alpha_{\ell-1})) \cap \Omega \neq \emptyset$.
\end{itemize}

\noindent
Observe that $\calF$ is monotone. Let $\gfix(\calF)$ be the greatest fixed-point of $\calF$. For every $(s,m) \in S_\Box \times \mem$, we put $\sel(s,m) = \Succ(s,m) \cap \gfix(\calF)$. Further, the set $\init$ consists of all $ (s,(0,0,\vec{0},\ldots\vec{0})) \in \gfix(\calF)$. It follows directly from the definition of $\Gamma$ that $\sel(s,m) \neq \emptyset$ for all $(s,m) \in \reach(\init)$ such that $s \in S_\Box$.
{\tiny }

Since $\gfix(\calF)$ can be computed in $\calO(|S|^2 \cdot |E| \cdot W^2 \cdot (\maxr{} \cdot W)^{2k\cdot (W/D)})$ time by the standard iterative algorithm, the strategy scheme $\Gamma = (\mem,\up,\sel,\init)$ can also be computed in this time. Further, observe the following:
\begin{itemize}
  \item[(A)] Let $(s_0,m_0),(s_1,m_1),\ldots$ be a run in $\game_\Gamma$ such that $(s_0,m_0) \in \init$. Then $s_0,s_1,\ldots$ is a run in $\game$ that satisfies the window-stability objective $\Phi$.
  \item[(B)] Let $(s,m) \not\in \gfix(\calF)$, and let $\Gamma^*$ be a strategy scheme which is the same as $\Gamma$ except for its constrainer $\sel^*$ which is defined by $\sel^*(s,m) = \Succ(s,m)$ for all $(s,m) \in S_\Box \times \mem$. Then there is a strategy $\pi^* \in \Pi_{\game_{\Gamma^*}}$ such that for every strategy $\sigma^* \in \Sigma_{\game_{\Gamma^*}}$ we have that $\outcome{\sigma^*}{\pi^*}{(s,m)}$  visits a configuration $(t,(D-1,\ell-1,\alpha_0,\ldots,\alpha_{\ell-1}))$ where $\alpha_{\ell-1} + \varrho(t) < \mu\cdot W$ or $\alpha_{\ell-1} + \varrho(t) > \nu\cdot W$.
\end{itemize}

\noindent
Both (A) and (B) follow directly from the definition of $\calF$. Now we can easily prove that $\Gamma$ indeed encodes the window-stability objective $\Phi$, i.e., $\Sigma_{\game}^{\Phi}(s) \approx_s \Sigma_{\game_\Gamma}(s)$ for all $s \in S$. 

Let $\tau \in \Sigma_{\game_\Gamma}(s)$. We need to show that $\sigma[\tau,s]$ achieves the objective~$\Phi$ in~$s$. So, let $\pi \in \Pi_\game$, and let $s_0,s_1,\ldots$ be the run $\outcome{\sigma[\tau,s]}{\pi}{s}$. Obviously, there is a corresponding run $(s_0,m_0),(s_1,m_1),\ldots$ in $\game_\Gamma$ initiated in $(s,\init(s))$, which means that $s_0,s_1,\ldots$ satisfies $\obj$ by applying~(A). 

Now let $\sigma \in \Sigma_{\game}^{\Phi}(s)$. We need to show that $\sigma$ is admitted by $\Gamma$ in~$s$. Suppose it is not the case. If $\init(s) = {\perp}$, then $(s,(0,0,\vec{0},\ldots,\vec{0})) \not\in \gfix(\calF)$, and hence $\sigma \not\in \Sigma_{\game}^{\Phi}(s)$ by applying~(B). If $\init(s) \neq {\perp}$, there is a finite path $s_0,\ldots,s_n,s_{n+1}$ of \emph{minimal length} such that 
$s_0 = s$, $s_n \in S_\Box$, and the corresponding finite path $(s_0,m_0),\ldots,(s_n,m_n),(s_{n+1},m_{m+1})$ in $\game_{\Gamma^*}$, where $m_0 = \init(s)$ and $m_{i+1} = \up(s_i,m_i)$ for all $0\leq i \leq n$, satisfies that $s_{n+1} \not\in \sel(s_n,m_n)$. Note that for all $s_i \in S_\Box$ where $i<n$ we have that $s_{i+1} \in \sel(s_i,m_i)$, because otherwise we obtain a contradiction with the minimality of $s_0,\ldots,s_n,s_{n+1}$.  Since $(s_{n+1},m_{n+1}) \not\in \gfix(\calF)$, by applying~(B) we obtain a strategy $\pi^* \in \Pi_{\game_{\Gamma^*}}$ such that for every $\sigma^* \in \Sigma_{\game_{\Gamma^*}}$ we have that $\outcome{\sigma^*}{\pi^*}{(s_{n+1},m_{n+1})}$  visits a configuration $(t,(D-1,\ell-1,\alpha_0,\ldots,\alpha_{\ell-1}))$ where $\alpha_{\ell-1} + \varrho(t) < \mu\cdot W$ or $\alpha_{\ell-1} + \varrho(t) > \nu\cdot W$. Let $\pi \in \Pi_\game$ be a strategy satisfying the following conditions: 
\begin{itemize}
	\item $\outcome{\sigma}{\pi}{s}$ starts with $s_0,\ldots,s_{n+1}$. 
	\item For all finite paths of the form $s_0,\ldots,s_{n+1},\ldots ,s_t$ in $\game$ such that $s_t \in S_\Diamond$, let $(s_0,m_0),\ldots,(s_{n+1},m_{n+1}),\ldots,(s_t,m_t)$ be the unique corresponding finite path in $\game_{\Gamma^*}$. We put $\pi(s_0,\ldots,s_n,\ldots s_t) = s_{t+1}$, where $\pi^*((s_{n+1},m_{n+1}),\ldots,(s_t,m_t)) = (s_{t+1},m_{t+1})$.  
\end{itemize}

\noindent
Clearly, the run $\outcome{\sigma}{\pi}{s}$ does \emph{not} satisfy the objective~$\Phi$, which contradicts the assumption $\sigma \in \Sigma_{\game}^{\Phi}(s)$.
\end{proof}

For every window-stability multi-objective $\Delta = \obj_1 \wedge \cdots \wedge \obj_n$ where $\obj_i = (W_i,D_i,\varrho_i,\mu_i,\nu_i)$, we put $M_{\Delta} = \prod_{i=1}^n W_i \cdot (\maxr{i}\cdot W_i)^{k_{i} \cdot (W_i/D_i)}$, where $k_i = \dimr{i}$. As a direct corollary to Theorem~\ref{thm-selector} and Proposition~\ref{prop-sel-conj}, we obtain the following:

\begin{corollary}
\label{cor-window-permissive}
	Let $\game = (S, (S_\Box,S_\Diamond),E)$ be a finite-state game and $\Delta$ a window-stability multi-objective. Then there is a permissive strategy scheme for \mbox{$\Delta$} with $M_\Delta$ memory elements constructible in time $\calO(|S|^2 \cdot |E| \cdot M_{\Delta}^2)$. 
\end{corollary}

Now we can formulate a (meta)theorem about the solvability of objectives of the form $\Delta \wedge \psi$, where $\Delta$ is a window-stability multi-objective and $\psi$ is a reward-based objective such that the time complexity of solving $\Psi$ for a game $\game = (S, (S_\Box,S_\Diamond),E)$ and a reward function $\varrho$ can be asymptotically bounded by a function $f$ in $|S|$, $|E|$, $\maxr{}$, and $\dimr{}$. 

\begin{theorem}
\label{thm-meta-thm}
	Let $\Psi$ be a reward-based objective solvable in $\calO(f(|S|,|E|,\maxr{},\dimr{}))$ time for every finite-state game $\game = (S, (S_\Box,S_\Diamond),E)$ and every reward function $\varrho$ for~$\Psi$. Further, let $\Delta$ be a window-stability multi-objective. Then the objective $\Delta \wedge \Psi$ is solvable in time 
	\[ 
	  \calO(\max\{f(|S|\cdot M_\Delta, |E| \cdot M_\Delta, \maxr{}, \dimr{}), |S|^2 \cdot |E| \cdot M_\Delta^2\})
	\] for every finite-state game $\game = (S, (S_\Box,S_\Diamond),E)$ and every reward function $\varrho$ for~$\Psi$.
\end{theorem} 

\noindent
Note that Theorem~\ref{thm-meta-thm} is a simple consequence of Corollary~\ref{cor-window-permissive} and Proposition~\ref{prop-conj}. 

Since mean-payoff objectives are solvable in $\calO(|S|\cdot |E| \cdot \maxr{})$ time when $\dimr{} =1$ \cite{Brim2010} and in  $\calO(|S|^2 \cdot |E| \cdot \maxr{} \cdot k \cdot (k \cdot |S| \cdot \maxr{})^{k^2+2k+1})$ time when $\dimr{} = k \geq 2$ \cite{Chatterjee:2013:HST:2529453.2529498}, we finally obtain:

\begin{theorem}
	Let $\game = (S, (S_\Box,S_\Diamond),E)$ be a finite-state game, $\Delta$ a window-stability multi-objective, and $\Psi = (\varrho,b)$ a mean-payoff objective. If $\dimr{} = 1$, then the objective $\Delta \wedge \Psi$ is solvable in time $\calO(|S|^2\cdot |E| \cdot M_\Delta^2 \cdot \maxr{})$. If $\dimr{} = k \geq 2$, then the objective $\Delta \wedge \Psi$ is solvable in time $\calO(|S|^2 \cdot |E| \cdot M_\Delta^3 \cdot \maxr{} \cdot k \cdot (k \cdot |S| \cdot M_\Delta \cdot \maxr{})^{k^2+2k+1})$.
\end{theorem}

\noindent
Let us note that for a given window-stability multi-objective $\Delta$ and a given one-dimensional reward function $\varrho$, there exists the \emph{maximal} bound $b$ such that the objective $\Delta \wedge (\varrho,b)$ is achievable. Further, this bound $b$ is rational and computable in time $\calO(|S|^2\cdot |E| \cdot M_\Delta^2 \cdot \maxr{})$.

\subsection{Lower Bounds for Window-Stability Objectives}
\label{sec-hardness}

We now focus on proving lower bounds for solving the window-stability objectives. More precisely, we establish lower complexity bounds for the problem whose instances are triples of the form $(\game,s,\obj)$, where $\game$ is a game (or a graph), $s$ is a state of $\game$, $\obj= (W,D,\varrho,\mu,\nu)$ is a window-stability objective, and the question is whether there exists a strategy $\sigma \in \Sigma$ which achieves $\obj$ in~$s$. The components of $\obj$ can be encoded in unary or binary, which is explicitly stated when presenting a given lower bound.

The main result of this section is Theorem~\ref{thm:hardness-unary} which implies that solving a window-stability objective is $\PSPACE$-hard for games and $\NP$-hard for graphs even if $\dimr{} = 1$, $D=1$, and $W$ as well as the values 
$\varrho(s)$ for all $s \in S$ are encoded in~\emph{unary}. Note that an upper time complexity bound for solving these objectives is $\calO(|S|^2 \cdot |E| \cdot W \cdot (\maxr{} \cdot W)^{W/D})$ by Corollary~\ref{cor-window-permissive}. Hence, the parameter which makes the problem hard is $W/D$. 

As a warm-up, we first show that lower bounds for solving the window-stability objectives where the reward function is of higher dimension, or $W$, $D$, and the rewards are encoded in binary, follow rather straightforwardly from the literature. Then, we develop some new insights and use them to prove the main result.

\begin{theorem}
Solving the window-stability objectives (where $\dimr{}$ is not restricted) is $\EXPTIME$-hard. The hardness result holds even if 
\begin{enumerate}
\item
the problem is restricted to instances where each component of each reward vector is in $\{-1,0,1\}$, or
\item
the problem is restricted to instances where the reward vectors have dimension one (but the rewards are arbitrary binary-encoded numbers).
\end{enumerate}
\end{theorem}
\begin{proof}
The result can be proven by a straightforward adaptation of the proof of $\EXPTIME$-hardness of multi-dimensional fixed-window mean-payoff problem~\cite[Lemma~23 and~24]{raskin-windows}. The reductions in~\cite{raskin-windows} that we can mimic are from the acceptance problem for polynomial-space alternating Turing machines (item 1.) and \emph{countdown games}~\cite{JSL:timed-automata-countdown-games} (item 2.).
Although the fixed-window mean-payoff problem differs from ours (see Section~\ref{sec-intro}), an examination of the proofs in~\cite{raskin-windows} reveals that almost the same constructions work even in our setting. In particular, while the problem to which countdown games are reduced in~\cite{raskin-windows} assumes two-dimensional rewards, in our setting we can restrict to single dimension due to window-stability objective imposing both a lower and an upper bound on local mean payoff.
\end{proof}

The reductions in the previous theorem require that the window size $W$ is encoded in binary, as the windows need to be exponentially long in the size of the constructed graph. For the case when $W$ is given in unary encoding, the following result can be adapted from~\cite{raskin-windows}.

\begin{theorem}
\label{thm-hard-one}
Solving the window-stability objectives (where $\dimr{}$ is not restricted) where the window size $W$ is encoded in unary is $\PSPACE$-hard, even if it is restricted to instances where the components of reward functions are in $\{-1,0,1\}$.
\end{theorem}

\noindent
A proof of Theorem~\ref{thm-hard-one} is obtained by adapting a proof from~\cite[Lemma 25]{raskin-windows}, where a reduction from generalized reachability games is given.

The results of~\cite{raskin-windows} do not yield lower bounds for window-stability objectives with one-dimensional reward functions in which either the windows size or the rewards are %
encoded in unary. 
In our setting, for the case of binary rewards/unary window size one can come up with $\NP$-hardness for graphs and $\PSPACE$-hardness for games via reductions from the $\mathsf{Subset}$-$\mathsf{Sum}$ problem and its quantified variant~\cite{Travers:QSSum}, respectively. Similarly, for unary rewards/binary window size a $\PSPACE$-hardness for games via reduction from emptiness of 1-letter alternating finite automata~\cite{JS:1-L-AFA} seems plausible. We do not follow these directions, since we are able to prove even stronger and somewhat surprising result: solving window-stability objectives with one-dimensional reward functions is $\PSPACE$-hard for games and $\NP$-hard for graphs even if \emph{all} the numbers in the input instance are encoded in unary. The proof of this result requires a new proof technique sketched below.

We rely on reductions from special variants of the $\mathsf{SAT}$ and $\mathsf{QBF}$ problems. An instance of the $\mathsf{Balanced}$-$\mathsf{3}$-$\mathsf{SAT}$ problem is a propositional formula $\varphi$ in a 3-conjunctive normal form which contains an even number of variables. Such an instance is positive if and only if $\varphi$ admits a satisfying assignment which maps exactly half of $\varphi$'s variables to $1$ (\textit{true}). We can also define a quantified variant, a $\mathsf{Balanced}$-$\mathsf{QBF}$ problem: viewing a quantified Boolean formula $\psi = \exists x_1 \forall x_2 \cdots \exists x_{n-1} \forall x_n\, \varphi$ (where $\varphi$ is quantifier-free), as a game between player controlling existentially quantified variables (who strives to satisfy $\varphi$) and player controlling universal variables (who aims for the opposite), we ask whether the existential player can enforce assignment mapping exactly half of the variables to $1$ and satisfying $\varphi$ (a formal definition of $\mathsf{Balanced}$-$\mathsf{QBF}$ is given in Appendix~\ref{app:hardness}). The following lemma is easy.

\begin{lemma}
\label{lem:balanced-hardness}
The $\mathsf{Balanced}$-$\mathsf{QBF}$ problem is $\PSPACE$-complete. The $\mathsf{Balanced}$-$\mathsf{3}$-$\mathsf{SAT}$ is $\NP$-complete.
\end{lemma}

Let $\game$ be a finite-state game and  $\obj = (W,D,\varrho,\mu,\nu)$ a window-stability objective. An instance $(\game,s,\obj)$ is \emph{small} if $\dimr{}=1$, and $W$, $D$, $\maxr{}$, and the numerators and denominators of the fully reduced forms of $\mu$ and $\nu$, are bounded by the number of states of~$\game$.

\begin{theorem}
\label{thm:hardness-unary}
Solving the window-stability objectives with one-dimensional reward functions is $\PSPACE$-hard for games and \mbox{$\NP$-hard} for graphs, even if it is restricted to small instances.
\end{theorem}
\begin{proof}[Proof (sketch)]
We proceed by reductions from $\mathsf{Balanced}$-$\mathsf{3}$-$\mathsf{SAT}$ for graphs and from $\mathsf{Balanced}$-$\mathsf{QBF}$ for games. As the reductions are somewhat technical, we explain just their core idea. The complete reduction can be found in Appendix~\ref{app:hardness}.

Assume a formula $\varphi$ in 3-CNF with variables $\{x_1,\dots,x_n\}$, $n$ being even. Consider the graph $\game$ in Figure~\ref{fig:hardness-illustration}. Both the “upper” gadget (consisting of unprimed states) and the “lower” gadget (with primed states) represent a standard “assignment choice” gadget, in which Player $\Box$ selects an assignment to variables in $\varphi$ (e.g. choosing an edge going to $t_1$ from $s_1$ corresponds to setting variable $x_1$ to \emph{true} etc.). With no additional constraints, $\Box$ can choose different assignments in the two gadgets, and she may change the assignment upon every new traversal of the lower gadget. Now assign reward $1$ to states that correspond to setting some variable to \emph{true} and $0$ to all the other states, let window size $W=2n$, checkpoint distance $D=1$, $\mu=\frac{n}{2}$, and $\nu = \frac{n}{2}+\frac{1}{3n}$ (say). In order to satisfy the window-stability objective $(W,D,\rho,\mu,\nu)$ from $s_1$, $\Box$ has to select a balanced assignment in the upper gadget and moreover, mimic this assignment in all future points in the lower gadgets. The necessity of the first requirement is easy. For the second, assume that there is some $\ell$ such that in the \mbox{$\ell$-th} step of the outcome $\run$ the player chooses to go from, say, $s_i$ to $t_i$ (or from $s_i'$ to $t_i'$), while in the $(\ell+2n)$-th step she goes from $s_i'$ to $f_i'$. Then the rewards accumulated within windows starting in the $\ell$-th and $(\ell+1)$-th step, respectively, differ by exactly one. Thus, $|\lmp_{W,D,\ell}(\run) - \lmp_{W,D,\ell+1}(\run)|=1/2n> 1/3n$, which means that the local mean payoffs at the $\ell$-th and $(\ell+1)$-th checkpoint cannot both fit into the interval $[\mu,\nu]$.  
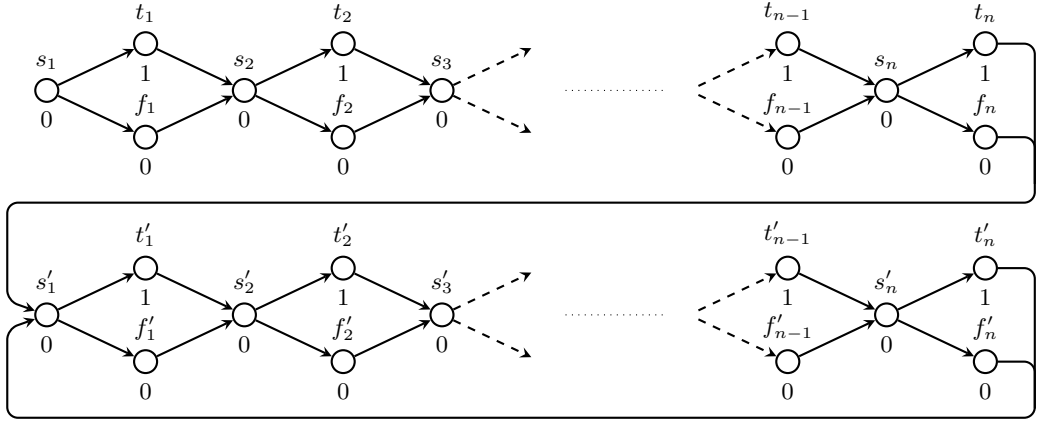
\begin{figure}[t]
\begin{tikzpicture}[x=1.3cm,y=0.62cm,font=\small]
\foreach \xc/\yc/\nnum/\ind in {0/0/1/1,2/0/2/2,4/0/3/3,8.5/0/4/n}  {%
	\node[label=below:$0$, label=above:$s_{\ind}$] (s\nnum) at (\xc,\yc) [ran] {};
}
\foreach \xc/\nnum/\ind in {1/1/1,3/2/2,7.5/3/n-1,9.5/4/n} {%
    \node[label=below:$1$, label=above:$t_{\ind}$] (t\nnum) at (\xc,1) [ran] {};
     \node[label=below:$0$, label=above:$f_{\ind}$] (f\nnum) at (\xc,-1) [ran] {};
}
\foreach \num/\xc in {3/5}{%
\node (td\num) at (\xc,1) {};
\node (fd\num) at (\xc,-1) {};
\begin{scope}[shift = {(0,-4.8)}]
\node (tdp\num) at (\xc,1) {};
\node (fdp\num) at (\xc,-1) {};
\end{scope}
}
\node (sd3) at (6.5,0) {};
\node (sdp3) at (6.5,-4.8) {};
\begin{scope}[shift = {(0,-4.8)}]
\foreach \xc/\yc/\nnum/\ind in {0/0/1/1,2/0/2/2,4/0/3/3,8.5/0/4/n}  {%
	\node[label=below:$0$, label=above:$s_{\ind}'$] (sp\nnum) at (\xc,\yc) [ran] {};
}
\foreach \xc/\nnum/\ind in {1/1/1,3/2/2,7.5/3/n-1,9.5/4/n} {%
    \node[label=below:$1$, label=above:$t_{\ind}'$] (tp\nnum) at (\xc,1) [ran] {};
     \node[label=below:$0$, label=above:$f_{\ind}'$] (fp\nnum) at (\xc,-1) [ran] {};
}
\end{scope}
\foreach \num in {1,2,4} {%
\draw [tran] (s\num) to (t\num);
\draw [tran] (sp\num) to (tp\num);
\draw [tran] (s\num) to (f\num);
\draw [tran] (sp\num) to (fp\num);
}
\foreach \src/\tgt in {1/2,2/3,3/4} {%
\draw [tran] (t\src) to (s\tgt);
\draw [tran] (f\src) to (s\tgt);
\draw [tran] (tp\src) to (sp\tgt);
\draw [tran] (fp\src) to (sp\tgt);
}
\foreach \num in {3}{%
\draw [tran,dashed] (s\num) to (td\num);
\draw [tran,dashed] (sp\num) to (tdp\num);
\draw [tran,dashed] (s\num) to (fd\num);
\draw [tran,dashed] (sp\num) to (fdp\num);
}
\draw [tran,dashed] (sd3) to (t3);
\draw [tran,dashed] (sdp3) to (tp3);
\draw [tran,dashed] (sdp3) to (fp3);
\draw [tran,dashed] (sd3) to (f3);

\draw [dotted] (5.25,0) to (6.25,0) {};
\draw [dotted] (5.25,-4.8) to (6.25,-4.8) {};
\draw [tran, rounded corners] (t4) -- ++(0.5,0) -- ++(0,-3.4) -- (-0.4,-2.4) -- (-0.4,-2.4 |- -0.4,-4.5) -- (sp1);
\draw [-,thick, rounded corners] (f4) -- ++(0.5,0) -- ++(0,-1);
\draw [tran, rounded corners] (tp4) -- ++(0.5,0) -- ++(0,-3.2) -- (-0.4,-7) -- (-0.4,-7 |- -0.4,-5.1) -- (sp1);
\draw [-,thick, rounded corners] (fp4) -- ++(0.5,0) -- ++(0,-1);
\end{tikzpicture}
\caption{In the lower gadget, Player $\Box$ must mimic the assignment she chose in the upper one.}
\label{fig:hardness-illustration}
\end{figure}

Note that we use the balanced variant of $\mathsf{3}$-$\mathsf{SAT}$ and $\mathsf{QBF}$, as to set up $\mu$ and $\nu$ we need to know in advance the number of variables assigned to \emph{true}.

Once we force the player to commit to some assignment using the above insight, we can add more copies of the “primed” gadget that are used to check that the assignment satisfies $\varphi$. Intuitively, we form a cycle consisting of several such gadgets, one gadget per clause of $\varphi$, the gadgets connected by paths of suitable length (not just by one edge as above). In each clause-gadget, satisfaction of the corresponding clause $C$ by the chosen assignment is checked by allowing the player to accrue a small additional reward whenever she visits a state representing satisfaction of some literal in $C$. This small amount is then subtracted and added again on a path that connects the current clause-gadget with the next one: subtracting forces the player to satisfy at least one literal in the previous clause-gadget (and thus accrue the amount needed to “survive” the subtraction) while adding ensures that this “test” does not propagate to the next clause-gadget. Rewards have to be chosen in a careful way to prevent the player from cheating. For $\PSPACE$-hardness of the game version we simply let the adversary control states in the initial gadget (but not in clause-gadgets) corresponding to universally quantified variables. %
\end{proof}

\section{The variance-stability problem}
\label{sec-results-variance}

In this section, we prove the results about variance-stability objectives promised in Section~\ref{sec-intro}. %

\begin{theorem}\label{thm:variance-in-NP}
The existence of a strategy achieving a given one-dimensional variance-stability objective for a given state of a given graph is in $\NP$. Further,
the strategy may require infinite memory. 
\end{theorem}
\begin{proof}
The proof is based on adaptation of techniques for the hybrid variance from~\cite{BCFK13} to the non-stochastic setting. 
In particular, we consider the graph as a Markov decision process. Subsequently, we apply results of~\cite{BCFK13} and reduce variance-stability problem to the problem of finding an appropriate solution of a negative semi-definite program, which belongs to $\NP$.
As part of the proof we show how to construct strategies that move through edges with the frequencies determined by the program. The proof is considerably complicated by the fact that the frequencies may be irrational and thus further limit processes are needed. Details can be found in Appendix~\ref{app:variance}.
\end{proof}
We now show that variance-stability objectives may require strategies with infinite memory. Consider the graph in Figure~\ref{fig-variance-counterexample}, and the variance-stability objective which requires to achieve the mean payoff at least $3/2$ and long-run variance at most $9/4$. 

Observe that there is an infinite-memory strategy solving the variance-stability which works as follows: We start in the state $A$, the strategy proceeds in infinitely many phases. In the $n$-th phase it goes $n$ times from $A$ to $B$ and back. Afterwards it goes to $D$, makes $2n$ steps on the loop on $D$, and then returns back to $A$. One can easily show, using the same technique as in the proof of Theorem~\ref{thm:variance-in-NP}, that the mean payoff converges along this run. The limit is obviously $4/2+0/2 + 1/2=\frac{3}{2}$ since the $-10$ reward is obtained with zero frequency. The long-run variance is 
\(
\frac{1}{4}\left(-\frac{3}{2}\right)^2 + \frac{1}{4}\left(4-\frac{3}{2}\right)^2 + \frac{1}{2}\left(1-\frac{3}{2}\right)^2=\frac{9}{4}.
\)
Now we show that there is no finite-memory strategy achieving the mean payoff $3/2$ and the long-run variance $9/4$. Note that the maximal mean payoff achievable (without any constraints) in the graph is $2$. %

Assume that there is a finite memory strategy $\sigma$ yielding mean payoff $x$ with $3/2 \le x \le 2$, and variance at most $9/4$. We first argue that $\sigma$ visits $C$ with zero frequency.
Denote by $f_Y$ the frequency of state $Y$.
Because $x = 0\cdot f_A + 4\cdot f_B + (-10)\cdot f_C + 1\cdot f_D$ by the definition of mean payoff, and also $f_A = f_B$ and $f_D = 1 - f_A -f_B - f_C$ by the definition of our graph,
we have $f_A = (x + 11\cdot f_C - 1)/2$ and $f_D = 2 - x - 12\cdot f_C$.
Thus, the variance is
{\small%
\begin{align*}
 f_A \cdot &(0-x)^2 + f_B \cdot (4-x)^2 + f_C \cdot(-10 - x)^2 + f_D (1 -x)^2\\
  &= \frac{x-1}{2} \cdot \big((0-x)^2 + (4-x)^2\big) + (2 - x) \cdot(1 -x)^2 \\
  &\qquad 
  + f_C\cdot \Big(\frac{11}{2}\cdot \big((0-x)^2 + (4-x)^2\big) + (-10 - x)^2 - 12 \cdot(1 -x)^2\Big).
\end{align*}}%
Using calculus techniques one can easily show that the first term is at least $9/4$ for all $x\in[3/2,2]$, while the parenthesized expression multiplied by $f_C$ is positive for all such $x$. Hence $f_C=0.$ But any finite-memory strategy that
stays in $C$ with frequency $0$ either eventually loops on $D$, in which case the mean payoff is only $1$, or it eventually loops on $A$ and $B$, in which case the
variance is $4$.

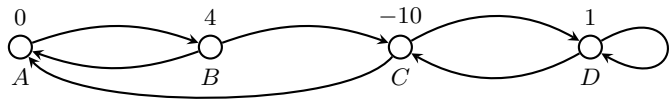
\begin{figure}
\begin{center}
	\begin{tikzpicture}[x=1cm,y=1cm,font=\small]
	\node[label=above:$4$, label=below:$B$] (B) at (0,0) [ran] {};
	\node[label=above:$0$, label=below:$A$] (A) at (-2.5,0) [ran] {};
	\node[label=above:$1$, label=below:$D$] (D) at (5,0) [ran] {};
	\node[label=above:$-10$, label=below:$C$] (C) at (2.5,0) [ran] {};

	\draw[tran, bend left=20] (B) to (A);
	\draw[tran, bend left=20] (A) to (B);
	\draw[tran, bend left=20] (B) to (C);

	\draw[tran, bend left] (D) to (C);
	\draw[tran, bend left] (C) to (D);
	\draw[tran, out=-130, in=-50, looseness=0.5] (C) to (A);
	\draw[loop right] (D) to (D);
	\end{tikzpicture}
\end{center}
\caption{One player game in which there is an infinite-memory strategy $\sigma$ such that $\lr(\outcome{\sigma}{\pi}{\sinit}) \ge 3/2$ and
$\va(\outcome{\sigma}{\pi}{\sinit}) \le 9/4$ (here $\pi$ is the only ``trivial'' strategy of the environment).  However, there is no finite-memory $\sigma$ with this property.}
\label{fig-variance-counterexample}
\end{figure}

Even finite-memory strategies that would approximate the desired variance-stability (up to some $\varepsilon>0$ error) must behave in a rather peculiar way: Infinitely many times stay in $\{A,B\}$ for a large number of steps (depending on $\varepsilon$) and also stay in $C$ for a large number of steps.
Hence, if the strategy was applied to a real-life system, a user would observe two disjoint repeating phases, one with low mean payoff but high instability, and one with low stability and high mean payoff.

\bibliographystyle{abbrv}
\bibliography{papers,bibtex-petr}

\newpage
\appendix
\begin{center}
\LARGE \bf Technical Appendix
\end{center}

\section{Proofs for Section~\ref{sec-algorithm}}

\begin{refproposition}{prop-sel-conj}
	Let $\game = (S, (S_\Box,S_\Diamond),E)$ be a finite-state game, and $n \in \Nset$. Further, for every $1 \leq i \leq n$, let $\Gamma_i = (\mem_i,\up_i,\sel_i,\init_i)$ be a strategy scheme for $\game$ which is permissive for $\obj_i$. Then there is a strategy scheme for $\game$ with $\prod_{i=1}^n |\mem_i|$ memory elements computable in $\calO(|S|^2 \cdot |E| \cdot \prod_{i=1}^n |\mem_i|^2)$ time which is permissive for $\obj_1 \wedge \cdots \wedge \obj_n$.
\end{refproposition}
\begin{proof}
	Intuitively, a strategy scheme $\Gamma = (\mem, \up, \sel, \init)$ which is permissive for $\obj_1 \wedge \cdots \wedge \obj_n$ is obtained as a ``synchronized product'' of all $\Gamma_i$. We put
	\begin{itemize}
		\item $\mem = \mem_1 \times \cdots \times \mem_n$
	\end{itemize}

    \noindent
    and define the memory update function by
    \begin{itemize}
    	\item $\up(s,(m_1,\ldots,m_n)) = (\up_1(s,m_1),\ldots,\up_n(s,m_n))$.
    \end{itemize}
    
    \noindent
    The constrainer requires more care. One might be tempted to define $\sel(s,(m_1,\ldots,m_n))$ as the intersection of all $\sel_i(s,m_i)$. However, this intersection may be empty, and hence must identify all elements of $S \times \mem$ which can reach such ``problematic'' $(s,(m_1,\ldots,m_n))$. Therefore, we define a function 	$\calF : 2^{S \times \mem} \rightarrow 2^{S \times \mem}$ as follows: For a given $\Omega \subseteq S \times \mem$, the set $\calF(\Omega)$ consists of all $(s,m) \in S \times \mem$ satisfying one of the following conditions:
    \begin{itemize}
    	\item $s \in S_\Box$ and there exists $s' \in  \bigcap_{i=1}^n \sel_i(s,m_i)$ such that $(s',\up(s,m)) \in \Omega$, where $m = (m_1,\ldots,m_n)$;
    	\item $s \in S_\Diamond$ and for all $(s,s') \in E$ we have that $(s',\up(s,m)) \in \Omega$.
    \end{itemize}
    
    \noindent
    Let $\gfix(\calF)$ be the greatest fixed-point of $\calF$. The constrainer $\sel$ is defined as follows:
    \begin{itemize}
    	\item $\sel(s,(m_1,\ldots,m_n))$ consists of all $s' \in S$ such that $s' \in \bigcap_{i=1}^n \sel_i(s,m_i)$ and $(s',\up(s,(m_1,\ldots,m_n))) \in \gfix(\calF)$.
    \end{itemize}
    
    \noindent
    Finally, we put
    \begin{itemize}
    	\item $\init = \{(s,(m_1,\ldots, m_n)) \mid (s,m_i) \in \init_i(s) \mbox{ for all } 1 \leq i \leq n \} \cap \gfix(\calF)$.
    \end{itemize}

	\noindent
	Observe that $\calF$ is monotone, and $\gfix(\calF)$ is computable in $\calO(|S|^2 \cdot |E| \cdot \prod_{i=1}^n |\mem_i|^2)$ time by the standard iteration algorithm. Hence, $\Gamma$ is also computable in  $\calO(|S|^2 \cdot |E| \cdot \prod_{i=1}^n |\mem_i|^2)$ time.  It is easy to check that $\Gamma$ is permissive for $\obj_1 \wedge \cdots \wedge \obj_n$.
\end{proof}

\section{Proofs for Section~\ref{sec-hardness}}
\label{app:hardness}

First we formally define the $\mathsf{Balanced}$-$\mathsf{QBF}$ problem. Let $\psi = Q_1 x_1\cdots Q_n x_n\, \varphi$ be a quantified boolean formula in a 3-CNF prenex normal form (3-CNF-PNF), i.e. for each $1\leq i \leq n$ we have $Q_i \in \{\forall,\exists\}$ and $\varphi$ is a quantifier-free formula in 3-CNF containing only variables from $\{x_1,\dots,x_n\}$. A \emph{model} of $\psi$ is a rooted directed tree $\calT$ satisfying the following properties:
\begin{itemize}
\item nodes of $\calT$ are labelled by elements of $\{1,\dots,n+1\}$, the root being labelled by $1$, all leaves by $n+1$, and children of each node labelled $i\leq n$ being labelled by $i+1$;
\item edges in $\calT$ are labelled by truth values $0$, $1$;
\item if $Q_i = \forall$, then all nodes labelled by $i$ have exactly two children, one connected via edge labelled by $0$ and the other by edge labelled by $1$; if $Q_i = \exists$, all nodes labelled by $i$ have exactly one child;
\item for each leaf $v$, the truth assignment induced by the unique path $p$ from the root to $v$ (i.e. $x_i$ is assigned the truth value labelling the unique edge on $p$ outgoing from an $i$-labelled node) satisfies $\varphi$.
\end{itemize}

\noindent
A model $\calT$ is \emph{balanced} if $n$ is even and all the assignments induced by root-leaf paths in $\calT$ assign $1$ to exactly half of the $n$ variables. In a $\mathsf{Balanced}$-$\mathsf{QBF}$ problem we are given a formula $\psi$ in a 3-CNF-PNF and we ask whether this formula admits a balanced model.

\begin{reflemma}{lem:balanced-hardness}
The $\mathsf{Balanced}$-$\mathsf{QBF}$ problem is $\PSPACE$-complete. The $\mathsf{Balanced}$-$\mathsf{3}$-$\mathsf{SAT}$ is $\NP$-complete.
\end{reflemma}
\begin{proof}
We prove the $\PSPACE$-completeness of the $\mathsf{Balanced}$-$\mathsf{QBF}$ problem, the result for the satisfiability variant can be obtained via identical reasoning. The membership in $\PSPACE$ can be easily tested by a non-deterministic polynomial-space Turing machine in the same manner as for the standard $\mathsf{QBF}$ problem. To prove the hardness result we employ a polynomial reduction from standard $\mathsf{QBF}$. Let $\psi=Q_1 x_1\cdots Q_n x_n\, \varphi$ be a formula in a 3-CNF-PNF. Clearly, $\psi$ is true iff it has a model. We construct a new formula $\psi'$ (again in 3-CNF-PNF) such that $\psi$ has a model iff $\psi'$ has a balanced model. We introduce new existentially quantified variables $x_{n+1},\dots,x_{2_n}$ and extend $\varphi$ with $n$ clauses of the form $x_i \vee \neg x_i$, $n+1\leq i \leq 2n$. Then every model $\calT$ of $\psi$ can be easily extended into balanced model of $\psi'$ as follows: for each leaf $v$ of $\calT$ we count the number $o$ of $1$'s on the path from the root to $v$. We then append to $v$ a path $p$ of length $n$ such that first $n-o$ edges on $p$ are labelled by $1$ and the remaining ones by $0$ (the nodes on $p$ are labelled according to the definition of a model). Conversely, every balanced model $\calT'$ of $\psi'$ can be pruned into a model of $\calT$ by simply removing all nodes with label $i> n+1$.
\end{proof}

\begin{reftheorem}{thm:hardness-unary}
	Solving the window-stability objectives with one-dimensional reward functions is $\PSPACE$-hard for games and \mbox{$\NP$-hard} for graphs, even if it is restricted to small instances.
\end{reftheorem}
\begin{proof}
We prove the $\PSPACE$-hardness for games, the result for graphs will then follow easily. We proceed by reduction from the $\mathsf{Balanced}$-$\mathsf{QBF}$ problem. Let $\psi = Q_1 x_1\cdots Q_n x_n\, \varphi$ be a quantified boolean formula in 3-CNF-PNF with $n$ even. We show how to construct, in polynomial time, a game $\game$, a state $s$ of $\game$, and a window-stability objective $\obj = (W,D,\varrho,\mu,\nu)$ where $\dimr{} = 1$ such that $(\game,s,\obj)$ is a small instance and $\psi$ admits a balanced model iff there is a strategy of player $\Box$ in game $\game$ achieving the stability objective $(W,D,\varrho,\mu,\nu)$ in the state $s$.

Let $\varphi=C_1 \wedge \cdots \wedge C_m$, where each $C_j$ is a clause. We put $W=2(n + m)$, $D=1$, $\mu = \frac{n}{2W}$, and $\nu = \frac{n}{2W}+\frac{1}{5W}$. We then construct the game $\game$ and the reward function $\varrho$ out of several gadgets: for each $0\leq j \leq m$ we have gadgets $\game_j^{\mathit{sat}}$ and $\game_j^{\mathit{force}}$. 

The gadget $\game_0^{\mathit{sat}}$ consists of states $s_0^i,t_0^{i,1},t_0^{\ i,0}$, $1\leq i \leq n$. A state $s_0^i$ belongs to Player $\Box$ iff $Q_i = \exists$, otherwise it belongs to $\Diamond$. All the other states %
belong to Player~$\Box$. For each $1\leq i \leq n$ there are edges $(s_0^i,t_0^{i,1}),(s_0^i,t_0^{i,0})$ and for $1\leq i <n$ we have edges $(t_0^{i,1},s_0^{i+1}),(t_0^{i,0},s_0^{i+1})$. States of the form $s_0^{i,1}$, $1\leq i \leq n$, receive reward $1$, all other states have reward $0$.

For each $1\leq j \leq m$, the gadget $\game_j^{\mathit{sat}}$ is defined similarly, but with certain differences: We have states $s_j^i,t_j^{i,1},t_j^{i,0}$, $1\leq i \leq n$, all of them belonging to $\Box$. Moreover, for each variable $x_k$ in $C_j$ we have Player $\Box$'s state $r_j^{k,z(j,k)}$, where $z(j,k)$ is $0$ or $1$ depending on whether $x_k$ appears negated in $C_j$ or not (if it appears bot negated and unnegated, any choice can be taken, for concreteness we put $z(j,k)=1$). For each $1\leq i < n$ we have edges $(s_j^i,t_j^{i,1}),(s_j^i,t_j^{i,0}),(t_j^{i,1},s_j^{i+1}),(t_j^{i,0},s_j^{i+1})$. Additionally we have edges $(s_j^n,t_j^{n,1}), (s_j^n,t_j^{n,0})$, and for each variable $x_k$ in $C_j$ an edge $(s_j^k,r_j^{k,z(j,k)})$ and, if $k<n$, an edge $(r_j^{k,z(j,k)},s_j^{k+1})$. Rewards are assigned as follows: states of the form $t_j^{i,1}$ for some $i$ have reward $1$, state of the form $r_j^{k,z(j,k)}$ has reward $\frac{11}{10}$ or $\frac{1}{10}$ depending on whether $z(j,k)=1$ or not, and all other states in $\game_j^{\mathit{sat}}$ have reward $0$.

We connect the gadgets $\game_j^{\mathit{sat}}$ with gadgets $\game_j^{\mathit{force}}$. For each $0\leq j \leq m$ the gadget $\game_j^{\mathit{force}}$ consists of states $u_j^{\ell}$, $1\leq \ell \leq 2m$ and edges $(u_j^{\ell},u_{j}^{\ell+1})$, $1\leq \ell < 2m$. We have
\[
\rho(u_j^{\ell}) = \begin{cases}
-\frac{1}{10} & \text{if $j>0$ and $\ell=2j-1$}\\
\frac{1}{10} & \text{if $j>0$ and $\ell=2j$}\\
0 & \text{otherwise}.
\end{cases}
\]
Finally, the connection of aforementioned gadgets is realized as follows: for each $0 \leq j \leq m$ we add edges $(t_j^{n,1},u_j^0)$, $(t_j^{n,0},u_j^0)$ and $(u_j^{2m},s_{(j \Mod{m}) + 1}^0)$. Moreover, for each $1\leq j \leq m$, if the above construction yields a state of the form $r_j^{n,z(j,n)}$, then we also add an edge $(r_j^{n,z(j,n)},u_j^0)$.

To prove the correctness of the reduction we employ additional handy notions. A path $s_1s_2\dots s_k \in S^*$ (here $S$ is the state set of $\game$) is an \emph{assignment path} if it does not contain cycles and moreover, for each $1\leq i \leq n$ it contains exactly one state from the set $\{t_j^{i,1},t_j^{i,0},r_j^{i,z(j,i)} \mid 0 \leq j \leq m\}$. An assignment path $h$ determines an assignment $\eta_h$ such that $\eta_h(x_i)=1$ if $h$ contains a state of the form $t_{j}^{i,1}$ or $r_j^{i,1}$ for some $j$, and $\eta_h(x_i)=0$ otherwise. 

We also extend the function $\rho$ to finite paths in $\game$: for $h=s_1,s_2,\dots,s_k$ we put $\rho(h)=\sum_{i=1}^{k}\rho(s_k)$. If $h$ is an acyclic path and $u_1,\dots,u_{k'}$ are all the states of the form $u_j^\ell$ contained in $h$, then we call the number $F(h)=\sum_{i=1}^{k'}\rho(u_i)$ a \emph{force}-reward of $h$. We also denote by $A(\run_k)$, $B(\run_k)$ and $C(\run_k)$ the numbers of occurrences of states of the form $t_j^{i,1}$, $r_j^{i,1}$ and $r_j^{i,0}$ in $\run_k$, respectively. Note that 
\begin{equation}
\label{eq:hardness-cost}
\rho(h) = A(h) + B(h) + (B(h) + C(h))/10 + F(h).
\end{equation}

Now let $\run=s_0,s_1,\dots$ be an arbitrary run in $\game$. For $k\in \Nset_0$ define a path $\run_k = s_k,s_{k+1},\dots,s_{k+2(n+m)}$.
The following claim easily follows from the construction of $\game$.

\begin{claim}
\label{cl:hardness}
 Let $\run$ be an arbitrary run in $\game$. Then for every $k\in \Nset_0$ the path $\run_k$ is an assignment path satisfying \emph{exactly} one of the following conditions
\begin{enumerate}
\item $\run_k$ contains states from exactly two distinct gadgets of the form $\game_j^{\mathit{sat}}$ and $F(\run_k)=0$; or 
\item $\run_k$ contains states from exactly on gadget of the form $\game_j^{\mathit{sat}}$ and $F(\run_k)\in [-\frac{1}{10},\frac{1}{10}]$. Moreover, in this case $\run_k$ contains, for each $1\leq i \leq n$ exactly one state from $\{t_j^{i,1},t_{j}^{i,0},r_{j}^{i,z(j,i)}\}$ and the state $s_j^i$.
\end{enumerate}
\end{claim}

\noindent
A run $\run$ in $\game$ is \emph{self-consistent} if for all $k,k'\in \Nset_0$ the assignments induced by the assignment paths $\run_k$ and $\run_{k'}$ are identical
A strategy $\sigma\in \Sigma_{\game}$ is \emph{self-consistent} if for all $\pi\in \Pi_\game$ the run $\outcome
{\sigma}{\pi}{s_0^1}$ is self-consistent.

Now assume that the formula $\psi$ admits a balanced model $\calT$. We define a strategy $\sigma$ in $\game$ as follows: for each finite path $h\in S^*S_{\Box}$ initiated in $s_0^1$ we have that
\begin{itemize}
\item if $h=s_0^1,t_0^{1,\ell_1},s_0^2,t_0^{2,\ell_2},\dots,t_{0}^{i-1,\ell_{i-1}},s_0^i$ for some $1\leq i \leq n$ and there is a node $v$ in $\calT$ such that $\ell_1,\dots,\ell_{i-1}$ is the sequence of labels of edges on the unique path from root of $\calT$ to $v$, then we put $\sigma(h)=t_0^{i,\ell}$, where $\ell$ is the label of the unique edge in $\calT$ outgoing from $v$;
\item if $h$ ends with a state of the form $s_j^i$, where $j>0$, then the length of $h$ is at least $2(n+m)$. The prefix of $h$ of length $2n$ is an assignment path inducing an assignment $\eta$. If $i$ is the smallest number in $\{1,\dots,n\}$ such that either $\eta(x_i)=1$ and $C_j$ contains literal $x_i$ or $\eta(x_i)=0$ and $C_j$ contains literal $\neg x_i$ (but not literal $x_i$), then we put $\sigma(h)=r_j^{i,\eta(x_i)}$; otherwise we put $\sigma(h)=t_j^{i,\eta(x_i)}$.
\item In all other cases there is only one outgoing edge from the last state of $h$, and $\sigma$ selects this edge.
\end{itemize}

\noindent
For $h$ not initiated in $s_0^1$ we can define $\sigma(h)$ arbitrarily. The second item in the definition of $\sigma$ ensures that $\sigma$ is self-consistent. Now let $\pi\in\Pi_\game$ be arbitrary and denote $\run=\outcome{\sigma}{\pi}{s_0^1}$. We need to prove that the local mean payoff in each checkpoint of $\run$ lies between $\mu$ and $\nu$, or equivalently, that for all $k\in \Nset_0$ it holds $\rho(\lambda_k) \in [\frac{n}{2},\frac{n}{2}+\frac{1}{5}]$. 
From the first item in the definition of $\sigma$ it follows that the assignment $\eta$ induced by $\lambda_0$ (and thus also by all $\run_k$, $k\geq 0$, by self-consistency) is balanced and satisfies $\varphi$, since it can be obtained by a root-leaf traversal of $\calT$. By~\eqref{eq:hardness-cost} it holds $\rho(\run_k) = A(\run_k) + B(\run_k) + (B(\run_k) + C(\run_k))/10 + F(\run_k)$. From balancedness of $\eta$ we get $A(\run_k) + B(\run_k) =n/2$. Now we distinguish two cases. Either $\lambda_k$ satisfies the first item in Claim~\ref{cl:hardness}. Then $F(\run_k)=0$ and $B(\run_k) + C(\run_k) \leq 2$, since  $\sigma$ visits at most one state of the form $r_j^{i,z(j,i)}$ in each of the gadgets $\game_j^{\mathit{sat}}$. In particular, $\frac{n}{2}\leq \rho(\run_k) \leq \frac{n}{2}+\frac{1}{5}$. Or $\lambda_k$ satisfies the first item in Claim~\ref{cl:hardness}. In this case case $B(\run_k) + C(\run_k) = 1$: $\run_k$ intersects just one $G_j^{\mathit{sat}}$ hence at most one state of the form $r_j^{i,z(i,j)}$, and at the same time it visits all states of the form $s_i^j$ and one successor of each such state. Since $\eta$ satisfies $C_j$, from the definition of $\sigma$ it follows that $\run_k$ contains at least one $r_j^{i,z(j,i)}$. Moreover, from Claim~\ref{cl:hardness} we get $-\frac{1}{10}\leq F(\run_k)\leq 1/10$. This again implies  $\frac{n}{2} \leq \rho(\run_k) \leq \frac{n}{2}+\frac{1}{5}$.

For the other direction, let $\sigma\in \Sigma_\game$ be a strategy $\sigma$ such that for all $\pi \in \Pi_\game$ the run $\outcome{\sigma}{\pi}{s_0^1}$ satisfies the window-stability objective $(W,D,\rho,\mu,\nu)$. We use $\sigma$ to inductively construct a tree $\calT$. We start with a single node labelled by $1$. In each step, we do the following for all leaves $v$ of the current tree: let $x_i$ be the label of $v$. If $Q_i =\forall$, we append to $v$ two children labelled by ${i+1}$, labelling one of the new edges by $0$ and the other by $1$. If $Q_i=\exists$, then let $h=s_0^1,t_0^{1,\ell_1},s_0^2,t_0^{2,\ell_2},\dots,t_{0}^{i-1,\ell_{i-1}},s_0^i$ be the unique path in $\game$ from $s_0^1$ to $s_0^i$ such that $\ell_1,\dots,\ell_{i-1}$ is the sequence of edge-labels on the unique path from root to $v$ in the current tree. We append to $v$ a child labelled by $i+1$ connected via an edge labelled by $x\in\{0,1\}$ s.t. $\sigma(h)=t_0^{i,x}$.

To show that $\calT$ is a balanced model of $\psi$ we need to prove that for all $\pi\in \Pi_\game$ the assignment $\eta_\pi$ induced by $(\outcome{\sigma}{\pi}{s_0^1})_0$ is balanced and satisfies $\varphi$. So fix any $\pi$ and denote $\run =\outcome{\sigma}{\pi}{s_0^1}$. We start with proving that \emph{each} $\run_k$, $k\in\Nset_0$, induces a balanced assignment. Assume that there is $k$ such that this is not true. Recall~\eqref{eq:hardness-cost} and note that $B(\run_k)+C(\run_k)\leq 6$ (due to Claim~\ref{cl:hardness} and since $\varphi$ is in 3-CNF), $F(\run_k)\in[-1/10,1/10]$ (Claim~\ref{cl:hardness}) and $A({\run_k})+B(\run_k)\in \{0,1,\dots,n\}\smallsetminus\{n/2\}$ (non-balancedness). From~\eqref{eq:hardness-cost} it follows that either $ \rho(\run_k)\leq \frac{n}{2}-\frac{3}{10}$ or $ \rho(\run_k)\geq\frac{n}{2}+\frac{9}{10}$ a contradiction with $\sigma$ satisfying the window-stability objective.

For satisfaction of $\varphi$ we first show 
that $\run$ is self-consistent. Assume the contrary. Then there is $k$ such that the balanced assignments induced by $\run_k$ and $\run_{k+1}$ differ. Let $s$ be the first state of $\run_k$ and $s'$ the last state of $\run_{k+1}$. Note that $|\rho(\run_k)-\rho(\run_{k+1})|=|\rho(s)-\rho(t)|$. Since $\run_k$ and $\run_{k+1}$ induce different assignments, it holds $|\rho(s)-\rho(t)|\geq 9/10$. But then at least one of the values $\rho(\run_k),\rho(\run_{k+1})$ falls outside of the interval $[\frac{n}{2},\frac{n}{2}+\frac{1}{5}]$, a contradiction with $\sigma$ satisfying the window-stability objective. 

Now we prove that the assignment $\eta_\pi$ induced by $\run_0$ (and thus by all $\run_k$) satisfies $\varphi$. Assume that there is a clause $C_j$ not satisfied by $\eta_\pi$, and let $k$ be the smallest index such that $\run_k$ begins with a state $u_{j-1}^{2j} $ (at least one such $k$ must exist due to construction of $\game$). Note that $\run_k$ ends with a state $u_j^{2j-1}$ and $F(\run_k) = -\frac{1}{10}$, and $\game_j^{\mathit{sat}}$ is the only gadget of the form $\game_{j'}^{\mathit{sat}}$ intersected by $\run_k$. Since we assume that the assignment induced by $\run_k$ does not satisfy $C_j$, $\run_k$ does not contain any state of the form $r_j^{i,z(j,i)}$. Hence, $B(\run_k)+C(\run_k)=0$ and $A(\run_k)=n/2$ (as $\eta_\pi$ is balanced). From~\eqref{eq:hardness-cost} it follows that $\rho(\run_k)=\frac{n}{2}-1/10$, a contradiction with $\sigma$ satisfying the window-stability objective. This finishes the proof.

Note that if all the quantifiers in $\psi$ are existential (i.e. $\psi$ is an instance of $\mathsf{Balanced}$-$\mathsf{3}$-$\mathsf{SAT}$), then $\game$ does not contain any states of Player $\Diamond$. This shows the $\NP$-hardness for graphs. %
\end{proof}

\section{Proofs for Section~\ref{sec-results-variance}}
\label{app:variance}
\begin{reftheorem}{thm:variance-in-NP}
The existence of a strategy achieving a given one-dimensional variance-stability objective for a given state of a given graph is in $\NP$. Further, a finite description of a strategy achieving the objective is computable in exponential time, and the strategy may require infinite memory. 
\end{reftheorem}

\noindent
Let us consider a game $\game=(S, (S_\Box,S_\Diamond),E)$ and an instance of the variance-stability problem determined by a reward function $\varrho$ together with a mean-payoff bound $b\in \Qset$ and a variance bound $c\in \Qset$. We assume that all outcomes are initiated in a fixed initial state $\bar{s}$.

A {\em frequency vector} is a tuple $\left(f_e\right)_{e\in E}\in [0,1]^{|E|}$
with $\sum_{e\in E} f_e=1$ and 
\[
\sum_{s':(s',s)\in E} f_{(s',s)}=\sum_{s':(s,s')\in E} f_{(s,s')}
\]
for all $s\in S$.
Now consider the following constraints:
\begin{equation}\label{eq:f_const}
\lr := \sum_{s\in S} f_s\cdot \varrho(s) \geq a
\end{equation}
\begin{equation}\label{eq:g_const}
\va := \sum_{s\in S} f_s\cdot (\varrho(s)-\lr)^2  \leq b
\end{equation}
Here $f_s=\sum_{(s',s)\in E} f_{(s',s)}$ for every $s\in S$.

As every single player game is a special case of a Markov decision
 process, we may invoke Proposition~5. of~\cite{BCFK13} and obtain the following lemma.
\begin{proposition}[\cite{BCFK13}]\label{prop:freqs-exist}
Assume that there is a solution to the given variance-stability problem.
Then there is a frequency vector $\left(f_e\right)_{e\in E}$ 
satisfying the inequalities (\ref{eq:f_const}) and (\ref{eq:g_const}).
All $e\in E$ satisfying $f_e>0$ belong to the same strongly connected component of $\game$ reachable from $\bar{s}$.
\end{proposition}
The above inequalities (\ref{eq:f_const}) and (\ref{eq:g_const}) can be turned into a negative semi-definite program, using techniques of~\cite{BCFK13}, and hence decided in non-deterministic polynomial time~\cite{DBLP:journals/ipl/Vavasis90}. To finish our algorithm, we need to show that a solution to the above inequalities can also be turned into a strategy which visits each $e\in E$ with the frequency $f_e$. 

Let $\lambda=s_0s_1\ldots$ be a run. Given $e\in E$ and $i\in \Nset$ we define
\[a^e_i(\lambda)=\begin{cases}
	1 & \text{ if } (s_i,s_{i+1}) = e\\
	0 & \text{ otherwise}
\end{cases}\]
\begin{lemma}\label{lem:frequency_to_strategy}
Suppose $\left(f_e\right)_{e\in E}$ is a frequency vector such that all $e\in E$ satisfying $f_e>0$ belong to the same strongly connected component reachable from the initial state $\bar{s}$. Then there is a strategy $\sigma_f$ with $\lim_{i\rightarrow\infty} \frac{\sum_{j=0}^i a^e_i(\lambda)}{i+1} = f_e$ for all $e\in E$, where $\lambda$ is the
outcome under $\sigma_f$ (initiated in $\bar{s}$).
\end{lemma}
\begin{proof}
Let us assume, w.l.o.g., that $\game$ itself is strongly connected.
If, $\left(f_e\right)_{e\in E}$ is rational and all edges $e$ satisfying $f_e>0$ induce a strongly connected graph, we may easily construct the strategy $\sigma_f$ as follows. 
We multiply all numbers $f_e$ with the least-common-multiple of their denominators and obtain a vector of natural numbers $f'_e$ that still satisfy the above flow equations. Now we may imagine the game as a multi-digraph, where each edge $e$ has the multiplicity $f'_e$. It is easy to show that the flow equations are exactly equivalent to existence of a directed Euler cycle. From this Euler cycle in the digraph we immediately get a cycle in our game which visits each edge exactly $f'_e$ times. By repeating the path indefinitely we obtain a run with the desired frequencies $f_e$ of edges.

Now consider a general frequency vector $\left(f_e\right)_{e\in E}$, i.e. the frequencies $f_e$ may be irrational and the graph induced by edges $e$ with $f_e>0$ does not have to be strongly connected. Then we may still approximate $\left(f_e\right)_{e\in E}$ it by a sequence of rational vectors $\left(f^i_e\right)_{e\in E}$, here $f^i_e\rightarrow f_e$ as $i\rightarrow \infty$, satisfying the flow equations (this follows from the fact that all frequency vectors satisfying the flow equations form a closed polyhedron). For each of the vectors we have a finite cycle $c_i$ which, when repeated indefinitely, gives the frequencies $\left(f^i_e\right)_{e\in E}$. We use $C_i^s$ for such infinite sequence, initiated in $s$, and, to make it easier to concatenate cycles, w.l.o.g. we suppose that for all $i$ we have $f^i_e > 0$.

Let $\varepsilon_0\varepsilon_1\ldots$ be a strictly decreasing sequence of numbers converging to $0$, with $\varepsilon_0\le 0.1$.
For all $i$, let $L_i$ be a number such that for all $L' \ge L_i$ and all $s$ we have
\[
 \frac{\sum_{j=0}^{L'} a^e_j(C_i^s)}{L'+1} \ge f^i_e - \varepsilon_i
\]

We define runs $\alpha_i$ and numbers $K_i$ as follows. We put $\alpha_1=C_1^{s_0}$ and $K_1 = 1$. 
Further, we define a $\alpha_i$ by taking $\alpha_{i-1}$ for $K_{i-1}$ steps, and then concatenating $C_i^s$, where $s$ is the $(K_{i-1}+1)$-th element of $\alpha_{i-1}$.
We let $K_i > K_{i-1}$ be a number such that
\[
 \frac{\sum_{j=0}^{K_i} a^e_j(\alpha_{i-1})}{K_i+1} \ge f^i_e - \varepsilon_i \text{\quad and\quad}\varepsilon_i^2 \cdot K_i \ge  L_{i+1}.
\]
Let $\alpha$ be the limit of the sequences $\alpha_i$ (note that $\alpha$ agrees with any $\alpha_i$ on the first $K_i$ steps).

We claim that, for all $e$,
\begin{equation}\label{eq:conv}
 \lim_{n\rightarrow\varepsilon} \frac{\sum_{j=0}^n a^e_j(\alpha)}{n+1} = f_e.
\end{equation}
Let $n$ be a number such that $K_i \le n \le K_{i+1}$.
We will show that, for all $e$,
\[
 \frac{\sum_{j=0}^n a^e_j(\alpha)}{n+1} \ge \min\{f^i_e,f^{i+1}_e\} - 2\cdot \varepsilon_i
\]
which by the convergence of $f^i$ and $\varepsilon_i$ will show \eqref{eq:conv}.
\begin{itemize}
\item If $n \le K_i + L_{i+1}$, then
 \begin{align*}
  \frac{\sum_{j=0}^n a^e_j(\alpha)}{n+1} &\ge \frac{\sum_{j=0}^n a^e_j(\alpha_i)}{n+1}
   \ge \frac{\sum_{j=0}^{K_i} a^e_j(\alpha_i)}{n+1}
   \ge \frac{K_i \cdot (f_e^i - \varepsilon_i)}{n+1}\\
   &\ge \frac{K_i \cdot (f_e^i - \varepsilon_i)}{K_i + L_{i+1}}
   \ge \frac{(K_i + L_{i+1}) \cdot (f_e^i - 2\cdot\varepsilon_i)}{K_i + L_{i+1}}
   \ge f_e^i - 2\cdot\varepsilon_i
 \end{align*}
 where the last-but-one inequality follows because
 \begin{align*}
  K_i \cdot (f_e^i - \varepsilon_i) &\ge (K_i +  L_{i+1}) \cdot (f_e^i - 2\cdot\varepsilon_i)\\
  K_i \cdot (f_e^i - \varepsilon_i) &\ge (K_i +  K_i\cdot \varepsilon_i^2) \cdot (f_e^i - 2\cdot\varepsilon_i)\\
   0 &\ge - K_i \cdot \varepsilon_i + K_i\cdot \varepsilon^2_i\cdot f_e^i - K_i\cdot 2\cdot \varepsilon^3_i\\
   0 &\ge K_i\cdot (-\varepsilon_i + \varepsilon^2_i\cdot f_e^i - 2\cdot \varepsilon^3_i)
 \end{align*}
 and $(-\varepsilon_i + \varepsilon^2\cdot f_e^i - 2\cdot \varepsilon^3_i)$ is negative as $\varepsilon_i\le 0.1$
\item
 On the other hand, if $K_i + L_{i+1} \le n \le K_{i+1}$, then
 \begin{align*}
  \frac{\sum_{j=0}^n a^e_j(\alpha)}{n+1} &\ge \frac{\sum_{j=0}^{K_i} a^e_j(\alpha) + \sum_{j=K_i+1}^{n} a^e_j(\alpha)}{n+1}\\
  &\ge \frac{(K_1+1) \cdot (f^i_e-\varepsilon_i) - (n+1-(K_1+1)) \cdot (f^{i+1}_e-\varepsilon_{i+1})}{n+1}\\
  &\ge \min\{f^i_e,f^{i+1}_e\}-\varepsilon_i
 \end{align*}\qedhere
\end{itemize}
\end{proof}

\end{document}